%% file: preprint.tex
\documentclass[preprint]{article}
\usepackage{style}
\usepackage{macros}
\usepackage{hyperref}
\usepackage{authblk}

\begin{document}

\title{The \lambdamucal{T}}

\author[1,2]{Herman Geuvers}
\author[1]{Robbert Krebbers}
\author[1]{James McKinna}

\affil[1]{Radboud University Nijmegen}
\affil[2]{Eindhoven University of Technology}

\renewcommand\Authands{ and }

\maketitle

\input{abstract.tex}

\input{introduction.tex}

\input{goedelsT.tex}
\input{muT.tex}

\input{arithmetic_cps.tex}
\input{arithmetic_confluence.tex}

\input{arithmetic_sn.tex}

\input{conclusions.tex}

\input{acknowledgments.tex}

\bibliographystyle{alpha}
\begin{small}

\end{small}
\end{document}

%% file: abstract.tex
\begin{abstract}
Calculi with control operators have been studied as extensions of 
simple type theory. Real programming languages contain datatypes, so 
to really understand control operators, one should also include these 
in the calculus. As a first step in that direction, we introduce 
\lambdamuT{}, a combination of Parigot's \lambdamucal{} and 
\GodelsTfull{}, to extend a calculus with control operators with a 
datatype of natural numbers with a primitive recursor. 

We consider the problem of confluence on 
raw terms, and that of strong normalization for the well-typed terms. 
Observing some problems with extending the proofs of Baba \etal{} and 
Parigot's original confluence proof, we provide new, and improved, proofs 
of confluence (by complete developments) and strong normalization (by 
reducibility and a postponement argument) for our system. 

We conclude with some remarks about extensions, choices, and prospects for 
an improved presentation. 
\end{abstract}

%% file: introduction.tex
\section{Introduction}
In pursuit, on the one hand, of a satisfactory equational theory of
\cbv{} \lambdacal{}, and on the other, of a means to interpret
the computational content of classical proofs, a variety of calculi
with control operators have been proposed. Few of these systems address 
the problem of how to incorporate primitive datatypes in
direct style, preferring instead to consider the usual Church
encoding of datatypes or else to analyze computation over datatypes via 
CPS-translations. 

In part this appears to arise because of the technical difficulty in
getting standard results such as confluence or strong normalization, and
their proof methods, either for classical calculi, or for simply-typed
calculi with datatypes, to extend to their combination. 

This paper introduces a new \lambdacal{} with control, \lambdamuT, in
which for example constructs for \catch\ and \throw\ may be
represented, which moreover has a basic datatype of natural numbers
with a primitive recursor, in the style of \GodelsTfull. We
demonstrate that it is possible to achieve a synthesis of classical
computation with datatypes with a conventional metatheory of typing
and reduction. To show how the system can be used in programming, we
give a simple example in~\ref{example:program}, where we define a function
that multiplies the first $n$ values of $f : \natTypeT \to\natTypeT$ and 
throws an exception as soon as it encounters the value $0$.

\subsection{Our approach}
Since Lafont's counterexample~\cite{girard1989}, it is well known that a calculus 
providing a general content to classical logic cannot be confluent. It only 
may become confluent if one adds an evaluation strategy (\cbn{} or \cbv{}). 
To define a calculus with control operators and datatypes we have therefore 
observed a tension between the \cbn{} features taken directly from Parigot's 
\lambdamucal{}, and the need to add certain \cbv{} features to obtain a system 
that is confluent and satisfies a normal form theorem (each closed term of type 
$\natTypeT$ is convertible to a numeral). The \lambdamucal{T} is therefore 
a \cbn{} system with strict evaluation on datatypes.
To avoid losing a normal form theorem, we could not make it a full \cbn{} 
system, and to avoid losing confluence we had to restrict the primitive 
recursor to only allow conversion when the numerical argument is a numeral.

Given these technical considerations, we were able to prove that
\lambdamuT{} satisfies subject reduction, has a normal form theorem,
is confluent and strongly normalizing. The last two proofs are
non-trivial because various niceties are required to make the standard
proof methods work.

Our confluence proof uses the notion of parallel reduction and defines
a complete development for each term. Surprisingly, it was difficult
to find a confluence proof for the original untyped
\lambdamucal{}. Baba, Hirokawa and Fujita~\cite{baba2001} have given a 
confluence proof for \lambdamu{} without the \mbox{$\to_{\mu\eta}$-rule}
($\mu\alpha.[\alpha]t \to t$ provided that $\alpha \notin \FCV
t$). Although they suggest how to extend parallel reduction for the
$\to_{\mu\eta}$-rule, they do not provide a formal definition of the
complete development nor a proof. Nakazawa~\cite{nakazawa2003} has
successfully carried out their suggestion for a
\cbv{} variant of \lambdamu{}, but does not use
the notion of complete development. Walter 
Py's PhD thesis~\cite{py1998} was the only place where we have found a complete proof
of confluence for \lambdamu{}. It uses Aczel's generalization of
parallel reduction~\cite{Aczel:CR} and a number of postponement arguments. 
In the present paper we extend the methodology
of~\cite{baba2001} to the case of \lambdamuT{}, which also includes
the $\to_{\mu\eta}$-rule.

Our strong normalization proof proceeds by defining relations 
$\to_A$ and $\to_B$ such that $\to\, =\, \to_{AB}\, \defined\, \to_A \cup \to_B$.
First we prove that $\to_A$ is strongly normalizing by the reducibility method. 
Secondly, we prove that $\to_B$ is strongly normalizing and that both reductions 
commute in a way that we can obtain strong normalization for $\to_{AB}$. The 
first phase is inspired by Parigot's proof of strong normalization for the 
\lambdamucal{}~\cite{parigot1997}.

\subsection{Related work}
The extension of simply typed lambda calculus with control operators
and the observation that these operators can be typed using the rules
of classical logic is originally due to Griffin~\cite{griffin1990} and
has lead to a lot of
research~\cite{parigot1992,parigot1993,degroote1994,rehof1994,berger1995,
	coquand1996,BarbaneraBerardi96,ariola2003,bakel2005},
by considering variations on the control operators, the underlying calculus or the
computation rules, or by studying concrete examples of the
computational content of proofs in classical logic. The
\lambdamucal{} of Parigot~\cite{parigot1992} has become a central
starting point for much research in this area.

The extension with datatypes, to make the calculus into a real
programming language with control operators, has not received so
much attention. We briefly summarize the research done
in this direction and compare it with our work. 

Murthy has defined a system with control operators, arithmetic,
products and sums in his PhD thesis~\cite{murthy1990}. His system uses
the control operators $\con{C}$ and $\con{A}$ (originally due 
to~\cite{griffin1990}) and the semantics of these operators is specified
by evaluation contexts rather than local reduction rules, as we do. So
his system does not really describe a {\em calculus} for datatypes
and control. Furthermore, Murthy mainly considers
CPS-translations to give an operational semantics of his system and did
not prove properties like confluence or strong normalization. 

Crolard and Polonowski have considered a version of \GodelsTfull{} with
products and \callcc~\cite{crolard2011}. As with Murthy,
the semantics is presented by CPS-translations instead of a direct
specification via a calculus.  Therefore properties like confluence
and strong normalization are trivial because they hold for the target
system already.

Barthe and Uustalu have worked on CPS-translations for inductive and
coinductive types~\cite{barthe2002}. Their work includes a system with
a primitive for iteration over the natural numbers and the control
operator $\Delta$.  Unfortunately only some properties of
CPS-translations are proven.

Rehof and S{\o}rensen have described an extension of the \lambdacal{_\Delta} 
with basic constants and functions~\cite{rehof1994}.
Unfortunately their extension is quite limited. For example the primitive 
recursor $\nrecT$ takes terms, rather than basic constants, as its arguments.
Their extension does not allow this, making it impossible to define $\nrecT$.

Parigot has described a second-order variant of his
\lambdamucal{}~\cite{parigot1992}. This system is very powerful,
because it includes all the well-known second-order representable
datatypes.  However, it suffers from the same weakness as \SystemF{},
namely poor computational efficiency (for example, an
$O(n)$-predecessor function). Also, as observed in~\cite{parigot1992,
  parigot1993}, this system does not ensure {\em unique representation
  of datatypes}. For example, there is no one-to-one correspondence
between natural numbers and closed normal forms of the type of Church
numerals.

There have been various investigations into concrete examples of
computational content of classical proofs. Coquand gives an overview
in his notes~\cite{coquand1996}. An earlier example is~\cite{berger1995},
where a binpacking problem is analyzed using proof
transformations. More recent work is by Makarov~\cite{makarov2006},
who takes Griffin's calculus and adds various rules to optimize the
extracted program.

If we look in particular at \GodelsTfull{}, Berger, Buchholz and
Schwichtenberg have described a form of program extraction from
classical proofs~\cite{berger2000}. Their method extracts a
term from a classical proof in which all computationally 
irrelevant parts are removed. To prove the correctness of their
approach they give a realizability interpretation. However, since
their target language is \GodelsTfull{}, extracted programs do not
contain control mechanisms.

Caldwell, Gent and Underwood have considered program extraction from
classical proofs in the proof assistant \NuPrl{}~\cite{caldwell2000}. 
In their work they extend \NuPrl{} with a proof
rule for Peirce's law and they associate \callcc{} to the extraction
of Peirce's law. Now, program extraction indeed results in a program
with control. The main focus of their work is on using
program extraction to obtain efficient search algorithms. The authors do not
prove any meta theoretical results so it is unclear whether their
approach is correct for arbitrary classical proofs.

\subsection{Outline}
The paper is organized as follows:
\begin{itemize}
\item Section~\ref{section:godelsT} recapitulates \GodelsTfull,
  fixing notation and conventions, together with the key normal form
  property.
\item Section~\ref{section:muT} introduces \lambdamuT, our
  \GodelsTfull{} variant of Parigot's \lambdamucal{} extended with a
  datatype of natural numbers with primitive recursor \(\nrecT\). We define
  the basic reduction rules, whose compatible closure defines computation in
  \lambdamuT. We show how to represent rules for a statically bound \catch\
  and \throw\ mechanism. We prove subject reduction, and the
  extended analogue of the normal form property. 
\item In Section~\ref{section:muT_cps}, we develop the corresponding
  CPS-translation for \lambdamuT{}, and show it preserves typing and
  conversion.
\item Section~\ref{section:muT_confluence} contains one of our two
  principal technical contributions: a direct proof of confluence on
  the raw terms of \lambdamuT, based on a novel analysis of
  complete developments.
\item In Section~\ref{section:muT_sn}, our second technical
  contribution is to prove SN for our calculus, using the reducibility
  method and a postponement argument.
\item We close with some conclusions and indications for further work,
  both in extending our system with a richer type system, and
  in investigating a fully-fledged \cbv{} version.
\end{itemize}

%% file: goedelsT.tex
\section{\texorpdfstring{\GodelsTfull{}}{Goedel's T}}
\label{section:godelsT}
\index{lambda-calculus-T@\lambdacal{\GodelsT}}
\index{Goedels T@\GodelsTfull{}|see{\lambdacal{\GodelsT}}}
\GodelsTfull{} (henceforth \lambdaT{}) was introduced by G\"odel to 
prove the consistency of Peano Arithmetic~\cite{sorensen2006}. It arises from 
$\lambda\Arrow$ by addition of a base type for natural numbers and a  
construct for primitive recursion.

\begin{definition}
\label{definition:godelsT_types}
The  \emph{types} of \lambdaT{} are built from a basic type
(the natural numbers) and a function type ($\to$) as follows.
\[
	\rho, \sigma, \tau \inductive \natTypeT \separator \sigma \to \tau
\]
\end{definition}

\begin{definition}
\label{definition:godelsT_terms}
The  \emph{terms} of the \lambdaT{} are inductively defined over an 
infinite set of  \emph{$\lambda$-variables $(x, y, \ldots)$} as follows.
\begin{flalign*}
t, r, s  \inductive &\ x \separator  \lm x:\rho.r \separator ts 
	\separator 0 \separator \sucT t \separator \nrecT_\rho\ r\ s\ t 
\end{flalign*}
Here, $\rho$ ranges over \lambdaT{}-types.
\end{definition}

As one would imagine, the terms $0$, $\sucT$ and $\nrecT$ denote zero, the successor 
function and primitive recursion over the natural numbers, respectively. 
We let $\FV t$ denote the set of free variables of $t$ and we
define the operation of capture avoiding substitution $\subst t x r$ of 
$r$ for $x$ in $t$ in the usual way.

\begin{convention}
\label{convention:variables}
Although a $\lambda$-abstraction and $\nrecT$ construct are annotated by
a type, we omit these type annotations when they are obvious or 
not relevant. Furthermore, we use the  \emph{Barendregt convention}. That is, 
given an expression, we may assume that bound variables are distinct from free 
variables and that all bound variables are distinct.
\end{convention}

\begin{definition}
The  \emph{derivation rules for \lambdaT{}} are as
shown in Figure~\ref{fig:GodelsT_typing}.
\begin{figure}[h!]
\centering
\subfloat[var]{
	\AXC{$x : \rho \in \Gamma$}
	\UIC{$\ljudg \Gamma x \rho$}
	\normalAlignProof
	\DisplayProof}	
\subfloat[lambda]{
	\AXC{$\ljudg {\Gamma, x : \sigma} t \tau$}
	\UIC{$\ljudg \Gamma {\lm x : \sigma.t} {\sigma \to \tau}$}
	\normalAlignProof
	\DisplayProof}
\subfloat[app]{
	\AXC{$\ljudg \Gamma t {\sigma \to \tau}$}
	\AXC{$\ljudg \Gamma s \sigma$}
	\BIC{$\ljudg \Gamma {ts} \tau$}
	\normalAlignProof
	\DisplayProof}

\subfloat[zero]{
	\AXC{$\ljudg \Gamma 0 \natTypeT$}
	\normalAlignProof
	\DisplayProof}
\subfloat[suc]{
	\AXC{$\ljudg \Gamma t \natTypeT$}
	\UIC{$\ljudg \Gamma {\sucT t} \natTypeT$}
	\normalAlignProof
	\DisplayProof}
\subfloat[nrec]{
	\AXC{$\ljudg \Gamma r \rho$}
	\AXC{$\ljudg \Gamma s {\natTypeT \to \rho \to \rho}$}
	\AXC{$\ljudg \Gamma t \natTypeT$}
	\TIC{$\ljudg \Gamma {\nrecT_\rho\ r\ s\ t} \rho$}
	\normalAlignProof
	\DisplayProof}	

\caption{The rules for typing judgments in \lambdaT{}.}
\label{fig:GodelsT_typing}
\end{figure}
\end{definition}

\begin{definition}
\label{definition:GodelsT_reduction}
Reduction $t \to t'$ is defined as the 
compatible closure of the following rules. 
\begin{flalign*}
	(\lm x.t)r &\to \subst t x r \tag{$\beta$} \\
	\nrecT\ r\ s\ 0 &\to r \tag{$0$} \\
	\nrecT\ r\ s\ (\sucT t) &\to s\ t\ (\nrecT\ r\ s\ t) \tag{$\sucT$}
\end{flalign*}
As usual, $\tto$ denotes the reflexive/transitive closure and 
$=$ denotes the reflexive/symmetric/transitive closure.
\end{definition}

Although we do not specify a deterministic reduction strategy it is obviously possible to create a 
\cbn{} and \cbv{} version of \lambdaT{}. Yet it is 
interesting to remark that in a \cbv{} version of \lambdaT{}
calculating the predecessor takes at least linear time while in a \cbn{} 
version the predecessor can be calculated in constant time~\cite{colson1998}.

Fortunately, despite the additional features of \lambdaT{}, the important properties 
of $\lambda\Arrow$, subject reduction, confluence and strong normalization, 
are preserved~\cite{stenlund1972, girard1989}.

Because it is convenient to be able to talk about a term representing an actual natural 
number we introduce the following notation.

\begin{notation}
\index{Numerals in \lambdaT{}}
$\natenc n \defined \sucT^n 0$
\end{notation}

\begin{definition}
\index{Values!of \lambdaT{}}
 \emph{Values} are inductively defined as follows.
\[
	v,w \inductive 0 \separator \sucT v \separator \lm x.r
\]
\end{definition}

\begin{theorem}
Given a term $t$ that is in normal form and such that $\ljudg {} t \rho$:
\begin{enumerate}
\item If $\rho = \natTypeT$, then $t \equiv \natenc n$ for some $n \in \nat$.
\item If $\rho = \sigma \to \tau$, then $t \equiv \lm x.r$ for a variable $x$ and term $r$.
\end{enumerate}
\end{theorem}

As the following indicates, the system \lambdaT{} 
has quite some expressive power.

\begin{definition}
\index{Representable function}
A function $f : \nat^n \to \nat$ is  \emph{representable} in \lambdaT{} if
there is a term $t$ with $\ljudg {} t {\natTypeT^n \to \natTypeT}$ such that:
\[
	t\; \natenc{m_1} \ldots \natenc{m_n} = \natenc{f(m_1, \ldots, m_n)}
\]
\end{definition}

\begin{theorem}
\label{theorem:godelsT_representable}
The functions representable in \lambdaT{} are exactly the functions that are
provably recursive in first-order arithmetic\footnote{Here we are allowed to say either
Peano Arithmetic (PA) or Heyting Arithmetic (HA), because a function is provably
recursive in PA iff it is probably recursive in HA~\cite{sorensen2006}.}.
\end{theorem}

\begin{proof}
This is proven in~\cite{sorensen2006}.
\end{proof}

%% file: muT.tex
\section{The \texorpdfstring{\lambdamuT}{lambda-mu-T}-calculus}
\label{section:muT}
\index{Lambda-mu-T-calculus@\lambdamucal{T}}

In this section we present our \GodelsTfull{} extension of Parigot's 
\lambdamucal{} (henceforth \lambdamuT). 

\begin{definition}
\label{definition:muT_terms}
The  \emph{terms} and  \emph{commands} of \lambdamuT{} are mutually inductively 
defined over an infinite set of  \emph{$\lambda$-variables $(x, y, \ldots)$} and 
 \emph{$\mu$-variables $(\alpha, \beta, \ldots)$} as follows.
\begin{flalign*}
t, r, s  \inductive &\ x \separator  \lm x:\rho.r \separator ts \separator \mu \alpha:\rho.c
	\separator 0 \separator \sucT t \separator \nrecT_\rho\ r\ s\ t \\
c, d \inductive &\ [\alpha]t
\end{flalign*}
Here, $\rho$ ranges over \lambdaT{}-types (Definition~\ref{definition:godelsT_types}).
We give $[\alpha]t$ lower precedence than $sr$, 
allowing us to write $[\alpha]sr$ instead of $[\alpha](sr)$. 
\end{definition}

As usual, we let $\FV t$ and $\FCV t$ denote the set of free $\lambda$-variables 
and \mbox{$\mu$-variables} of $t$, respectively. Moreover, 
we  define substitution $\subst t x r$ of $r$ for $x$ in $t$, which
is capture avoiding for both $\lambda$- and $\mu$-variables, in the obvious way.
Similar to Convention~\ref{convention:variables}, we will often omit type annotations
for $\mu$-binders.

\begin{notation}
$\muthrow c \defined \mu \gamma.c$ provided that $\gamma \notin \FCV c$.
\end{notation}

\begin{definition}
\label{def:lambdamuT} The  \emph{typing rules for \lambdamuT{}} are as
shown in Figure~\ref{fig:muT_typing}.
\begin{figure}[h!]
\centering
\subfloat[axiom]{
	\AXC{$x : \rho \in \Gamma$}
	\UIC{$\mujudg \Gamma \Delta x \rho$}
	\normalAlignProof
	\DisplayProof}
\subfloat[lambda]{
	\AXC{$\mujudg {\Gamma, x : \sigma} \Delta t \tau$}
	\UIC{$\mujudg \Gamma \Delta {\lm x : \sigma.t} {\sigma \to \tau}$}
	\normalAlignProof
	\DisplayProof}
\subfloat[app]{
	\AXC{$\mujudg \Gamma \Delta t {\sigma \to \tau}$}
	\AXC{$\mujudg \Gamma \Delta s \sigma$}
	\BIC{$\mujudg \Gamma \Delta {ts} \tau$}
	\normalAlignProof
	\DisplayProof}
	
\subfloat[zero]{
	\AXC{$\mujudg \Gamma \Delta 0 \natTypeT$}
	\normalAlignProof
	\DisplayProof}
\subfloat[suc]{
	\AXC{$\mujudg \Gamma \Delta t \natTypeT$}
	\UIC{$\mujudg \Gamma \Delta {\sucT t} \natTypeT$}
	\normalAlignProof
	\DisplayProof}

\subfloat[nrec]{
	\AXC{$\mujudg \Gamma \Delta r \rho$}
	\AXC{$\mujudg \Gamma \Delta s {\natTypeT \to \rho \to \rho}$}
	\AXC{$\mujudg \Gamma \Delta t \natTypeT$}
	\TIC{$\mujudg \Gamma \Delta {\nrecT_\rho\ r\ s\ t} \rho$}
	\normalAlignProof
	\DisplayProof}	

\subfloat[activate]{
	\AXC{$\musjudg \Gamma {\Delta, \alpha : \rho} c$}
	\UIC{$\mujudg \Gamma \Delta {\mu \alpha : \rho.c} \rho$}
	\normalAlignProof
	\DisplayProof}
\subfloat[passivate]{
	\AXC{$\mujudg \Gamma \Delta t \rho$}
	\AXC{$\alpha : \rho \in \Delta$}
	\BIC{$\musjudg \Gamma \Delta {[\alpha]t}$}
	\normalAlignProof
	\DisplayProof}

\caption{The rules for typing judgments in \lambdamuT{}.}
\label{fig:muT_typing}
\end{figure}
\end{definition}

A typing judgment $\mujudg \Gamma \Delta t \rho$ is  \emph{derivable in
  \lambdamuT{}} in case it is the conclusion of a derivation tree
that uses the rules of Definition~\ref{def:lambdamuT}. We
say ``term $t$ has type $\rho$ in environment of
$\lambda$-variables $\Gamma$ and environment of
\mbox{$\mu$-variables} $\Delta$''.

Similarly, a typing judgment $\musjudg \Gamma \Delta c$ is {\em
  derivable in \lambdamuT{}\/} in case it is the conclusion of a
derivation tree that uses the rules of Definition~\ref{def:lambdamuT}. 
We say ``command $c$ is typable in 
environment of $\lambda$-variables $\Gamma$ and environment of
$\mu$-variables $\Delta$''.

\begin{fact} The typing judgment is closed under weakening of both environments.
That is, if $\mujudg \Gamma \Delta t \rho$, $\Gamma \subseteq \Gamma'$ and
$\Delta \subseteq \Delta'$, then $\mujudg {\Gamma'} {\Delta'} t \rho$.
\end{fact}

In order to define the reduction rules we first define the notions of 
 \emph{contexts} and  \emph{structural substitution}. Although the reduction rules 
merely require contexts of a restricted shape (those that are  \emph{singular}) we 
define contexts of a more general shape so we can reuse these definitions 
in our proof of confluence (Section~\ref{section:muT_confluence}) 
and strong normalization (Section~\ref{section:muT_sn}).

\begin{definition}
\label{definition:context}
\index{Context!call-by-name for \lambdamuT{}}
A  \emph{\lambdamuT{}-context} is defined as follows.
\[
	E \inductive \Box \separator E t \separator \sucT E \separator \nrecT\ r\ s\ E
\]
A context is  \emph{singular} if it is the following shape.
\[
	E^s \inductive \Box t \separator \sucT \Box \separator \nrecT\ r\ s\ \Box
\]
\end{definition}

\begin{definition}
Given a context $E$ and a term $s$,  \emph{substitution of $s$ 
for the hole in $E$}, notation $\cctx E s$, is defined as follows.
\begin{flalign*}
\cctx \Box s & \defined s\\
\cctx {(Et)} s & \defined \cctx E s t \\
\cctx {(\sucT E)} s & \defined \sucT \cctx E s \\
\cctx {(\nrecT\ r\ s\ E)} s & \defined \nrecT\ r\ s\ \cctx E s
\end{flalign*}
\end{definition}

\begin{definition}
Given contexts $E$ and $F$, the context $EF$ is defined by:
\begin{flalign*}
	\Box F &\defined F \\
	(Et)F & \defined (EF)t \\
	(\sucT E)F & \defined \sucT(EF) \\
	(\nrecT\ r\ s\ E)F & \defined \nrecT\ r\ s\ (EF)
\end{flalign*}
\end{definition}

\begin{fact}
$\cctx E {\cctx F {t}} \equiv \cctx {EF} t$
\end{fact}

Using contexts we can now define  \emph{structural substitution}.
Structural substitution of a $\mu$-variable $\beta$ and a context
$E$ for a $\mu$-variable $\alpha$ in $t$, notation $\subst t \alpha {\beta E}$, 
recursively replaces each command $[\alpha]q$ in $t$ by $[\beta]\cctx E {q'}$
where $q' \equiv \subst q \alpha {\beta E}$. Our notion of structural 
substitution is more general than Parigot's original 
presentation~\cite{parigot1992}. He defines $\subst t \beta \alpha$, 
which renames each $\mu$-variable $\beta$ in $t$ into $\alpha$, and $\subst t \alpha s$,
which replaces each command $[\alpha]q$ in $t$ by $[\alpha]q's$ where 
$q' \equiv \subst q \alpha s$. Of course, his
notions are just instances of our definition, namely, the former corresponds to 
$\subst t \beta {\alpha\ \Box}$ and the latter to $\subst t \alpha {\alpha\ (\Box s)}$.
Parigot's presentation suffices for the definition of the reduction
rules, but our presentation allows us to prove properties like confluence 
(Section~\ref{section:muT_confluence}) and strong normalization 
(Section~\ref{section:muT_sn}) in a more streamlined way.

\begin{definition}
\index{Structural substitution}
\label{lemma:mu_strucsubst}
 \emph{Structural substitution} $\subst t {\alpha} {\beta E}$ of a $\mu$-variable 
$\beta$ and a context $E$ for a $\mu$-variable $\alpha$
is defined as follows.
\begin{flalign*}
\subst x {\alpha} {\beta E} 
	&\defined x\\
\subst {(\lm x.r)} \alpha {\beta E} 
	&\defined \lm x.\subst r \alpha {\beta E} \\
\subst {(ts)} \alpha {\beta E} 
	&\defined \subst t \alpha {\beta E} \subst s \alpha {\beta E}\\
\subst 0 {\alpha} {\beta E} 
	&\defined 0\\
\subst {(\sucT t)} \alpha {\beta E} 
	&\defined \sucT (\subst t \alpha {\beta E})\\
\subst {(\nrecT\ r\ s\ t)} \alpha {\beta E} 
	&\defined \nrecT\ (\subst r \alpha {\beta E})\ (\subst s \alpha {\beta E})\ (\subst t \alpha {\beta E})\\
\subst {(\mu \gamma.c)} \alpha {\beta E} 
	&\defined \mu \gamma.\subst c \alpha {\beta E}  \\
\subst {([\alpha]t)} \alpha {\beta E} 
	&\defined [\beta] \cctx E {\subst t \alpha {\beta E}} \\
\subst {([\gamma]t)} \alpha {\beta E} 
	&\defined [\gamma]\subst t \alpha {\beta E} 
	\qquad \text{ provided that }\gamma \neq \alpha 
\end{flalign*}
Structural substitution is capture avoiding for both $\lambda$- and $\mu$-variables.
\end{definition}

\begin{definition}
Reduction $t \to t'$ is defined as the compatible 
closure of the following rules. 
\begin{flalign*}
	(\lm x.t)r &\to \subst t x r \tag{$\beta$}\\
	\sucT(\mu \alpha.c) &\to \mu \alpha.\subst c \alpha {\alpha\ (\sucT \Box)} \tag{$\mu \sucT$}\\
	(\mu \alpha.c)s &\to \mu \alpha.\subst c \alpha {\alpha\ (\Box s)} \tag{$\mu R$} \\
	\mu \alpha.[\alpha]t &\to t \qquad\text{ provided that }\alpha \notin \FCV t \tag{$\mu\eta$} \\
	{[\alpha]}\mu \beta.c &\to \subst c \beta {\alpha\ \Box} \tag{$\mu i$} \\
	\nrecT\ r\ s\ 0 &\to r \tag{$0$} \\
	\nrecT\ r\ s\ (\sucT \natenc n) &\to s\ \natenc n\ (\nrecT\ r\ s\ \natenc n) \tag{$\sucT$} \\
	\nrecT\ r\ s\ (\mu \alpha.c) &\to \mu \alpha.\subst c \alpha {\alpha\ (\nrecT\ r\ s\ \Box)} \tag{$\mu \natTypeT$}
\end{flalign*}
As usual, $\tto^+$ denotes the transitive closure, $\tto$ denotes the reflexive/transitive closure and 
$=$ denotes the reflexive/symmetric/transitive closure of $\to$.
\end{definition}

\begin{fact}
\label{fact:muT_red_singular_context}
As in~\cite{felleisen1992}, the notion of a singular 
context allows us to replace the reduction rules $\to_{\mu \sucT}$, $\to_{\mu R}$ and 
$\to_{\mu \natTypeT}$ by the following single rule.
\[
	\cctx {E^s} {\mu \alpha.c} \to  \mu \alpha.\subst c \alpha {\alpha E^s}
\]
\end{fact}

\begin{fact}
\label{fact:muT_lift_context}
$\cctx E {\mu \alpha.c} \tto  \mu \alpha.\subst c \alpha {\alpha E}$
\end{fact}

From a computational point of view one should think of $\mu\alpha.[\beta]t$ 
as a combined operation that catches exceptions labeled $\alpha$ in
$t$ and throws the results of $t$ to $\beta$. Following 
Crolard~\cite{crolard1999}, we define the operators 
$\catch$ and $\throw$.

\begin{definition}
The terms $\catchin \alpha t$ and $\throwto s \beta$ are defined as follows.
\begin{flalign*}
	\catchin \alpha t &\defined \mu \alpha.[\alpha]t \\
	\throwto s \beta &\defined \muthrow[\beta]s
\end{flalign*}
\end{definition}

Similar to commands, we give $\catchin \alpha t$ and $\throwto s \beta$ lower 
precedence than $sr$, allowing us to write $\catchin \alpha {sr}$ instead of 
$\catchin \alpha {(sr)}$.

Crolard~\cite{crolard1999} moreover defines a system with $\catch$ and $\throw$
as primitives
and proves a correspondence with the \lambdamucal{}. We prove
that the above simulation of $\catch$ and $\throw$ satisfies a generalization 
of Crolard's rules.

\begin{lemma}
We have the following reductions for \catch{} and \throw{}.
\begin{enumerate}
\item $\cctx E {\catchin \alpha t} \tto \catchin \alpha {\cctx E {\subst t \alpha {\alpha E}}}$
\item $\cctx E {\throwto t \alpha} \tto \throwto t \alpha$
\item $\catchin \alpha {\catchin \beta t} \to \catchin \alpha {\subst t \beta {\alpha\Box}}$
\item $\throwto {\throwto t \beta} \alpha \to \throwto t \beta$
\item $\throwto {\catchin \beta t} \alpha \to \throwto {\subst t \beta {\alpha\Box}} \alpha$
\item $\catchin \alpha {\throwto t \alpha} \to \catchin \alpha t$
\item $\catchin \alpha t \to t$ provided that $\alpha \notin \FCV t$
\end{enumerate}
\end{lemma}

\begin{proof}
These reductions follow directly from the reduction rules of \lambdamuT{},
except for (1) and (2) where we need Fact~\ref{fact:muT_lift_context}.
\end{proof}

The $\catch$ and $\throw$ as defined above give rise to a system with
 \emph{statically bound} exceptions. This is different from exceptions in 
for example \Lisp{}, where they are  \emph{dynamically bound}. In a system with 
dynamically bound exceptions, substitution is not capture avoiding for 
exception names.

\begin{example}
Consider the following term:
\[
	\catchin \alpha {\sucT ((\lambda f:\natTypeT\to\natTypeT\,.\,\catchin \alpha {f\,0})\ 
	  (\lambda x:\natTypeT\,.\,\throwto x \alpha))}.
\]
Here, both occurrences of $\catch$ bind different occurrences $\alpha$. 
So after two \mbox{$\beta$-reduction} steps we obtain
$\catchin \alpha {\sucT (\catchin \beta {\throwto 0 \alpha})}$ and hence
its normal form is $0$. In systems with dynamically bound exceptions this
term would reduce to $\sucT 0$ because the $\throw$ would get caught by 
the innermost $\catch$.
\end{example}

\begin{example}
\label{example:program}
We consider a simple \lambdamuT{}-program $F$ that, given $f : \natTypeT \to \natTypeT$, 
computes the product of the first $n$ values of $f$, that is
$F\,\natenc n = f\,0 * \ldots * f\,\natenc n$ for $n\in\nat$. The interest of this program is 
that it uses the exception mechanism to stop multiplying once a zero is
encountered. First we define addition and multiplication in the usual
way in \lambdamuT{}.
\begin{flalign*}
	(+) &\defined \lambda nm \,.\,\nrecT\ m\ (\lm x y\,.\,\sucT y)\ n \\
	(*) &\defined \lambda nm \,.\,\nrecT\ 0\ (\lm x y\,.\,m + y)\ n
\end{flalign*}
Now, given $f : \natTypeT \to \natTypeT$, we define the term 
$F \typeof \natTypeT\to\natTypeT$, using a `helper function' $H$, which does a 
case analysis on the value of $f\,y$, as follows. 
\begin{flalign*}
F &\defined \lm x\,.\,\catchin \alpha {\nrecT\ \natenc 1\ H\, (\sucT x)}\\
H &\defined \lm y\,m\,.\,\nrecT\ (\throwto 0 \alpha)\ 
	(\lm z\,\_\,.\,m * \sucT z)\ (f\,y).
\end{flalign*}

Let $f : \natTypeT \to \natTypeT$ be some term that satisfies $f\,\natenc{0} = \natenc{3}$,
$f\,\natenc{1} = \natenc{0}$ and $f\,\natenc{2} = \natenc{5}$. We show a
computation of $F\,\natenc 2$.
\begin{flalign*}
F\,\natenc 2 
	&\tto{} \catchin \alpha {\nrecT\ \natenc 1\, H\, \natenc 3}\\
	&\tto{} \catchin \alpha {H\,\natenc 2\,(\nrecT\ \natenc 1\,H\, \natenc 2)}\\
	&\tto{} \catchin \alpha {\nrecT\ (\throwto 0 \alpha)\ (\lm z\,\_\,.\,(\nrecT\ \natenc 1\ H\, \natenc 2) * \sucT z)\ (f\,\natenc 2})\\
	&\tto{} \catchin \alpha {(\nrecT\ \natenc 1\,H\, \natenc 2) * \natenc 5}\\
	&\tto{} \catchin \alpha {(H\,\natenc 1\,(\nrecT\ \natenc 1\,H\, \natenc 1)) * \natenc 5}\\
	&\tto{} \catchin \alpha {(\nrecT\ (\throwto 0 \alpha)\ (\lm z\,\_\,.\, (\nrecT\ \natenc 1\,H\, \natenc 1) * \sucT z)\ (f\,\natenc{1})) * \natenc 5}\\
	&\tto{} \catchin \alpha {\throwto 0 \alpha * \natenc 5}\\
	&\tto{} \catchin \alpha {\nrecT\ 0\ (\lm x y\,.\,\natenc 5 + y)\ (\throwto 0 \alpha)}\\
	&\tto{} \catchin \alpha {\throwto 0 \alpha}\\
	&\tto{} 0
\end{flalign*}
\end{example}

In order to prove that \lambdamuT{} satisfies subject reduction we have to
prove that each reduction rule preserves typing. Because some of the reduction
rules involve structural substitution it is convenient to prove an auxiliary 
result that structural substitution preserves typing. To express this
property we introduce the notion of a \emph{contextual typing judgment}, notation
\(\cjudg \Gamma \Delta E \rho \sigma\), which expresses that
$\mujudg \Gamma \Delta t \rho$ implies $\mujudg \Gamma \Delta {\cctx E t} \sigma$.

\begin{definition}
The derivation rules for the \emph{contextual typing judgment} 
\(\cjudg \Gamma \Delta E \rho \sigma\) are as shown in Figure~\ref{fig:muT_typing_context}.
\begin{figure}[h!]
\centering
\subfloat[hole]{
	\AXC{$\cjudg \Gamma \Delta \Box \rho \rho$}
	\normalAlignProof
	\DisplayProof {\hskip -4pt}}	
\subfloat[app]{
	\AXC{$\cjudg \Gamma \Delta E \rho {\sigma \to \tau}$}
	\insertBetweenHyps{\hskip -1pt}
	\AXC{$\mujudg \Gamma \Delta t \sigma$}
	\BIC{$\cjudg \Gamma \Delta {Et} \rho \tau$}
	\normalAlignProof
	\DisplayProof {\hskip -4pt}}
\subfloat[suc]{
	\AXC{$\cjudg \Gamma \Delta E \rho \natTypeT$}
	\UIC{$\cjudg \Gamma \Delta {\sucT E} \rho \natTypeT$}
	\normalAlignProof
	\DisplayProof}
	
\subfloat[nrec]{
	\AXC{$\mujudg \Gamma \Delta r \sigma$}
	\AXC{$\mujudg \Gamma \Delta s {\natTypeT \to \sigma \to \sigma}$}
	\AXC{$\cjudg \Gamma \Delta E \rho \natTypeT$}
	\TIC{$\cjudg \Gamma \Delta {\nrecT\ r\ s\ E} \rho \sigma$}
	\normalAlignProof
	\DisplayProof}
\caption{The rules for contextual typing judgments in \lambdamuT{}.}
\label{fig:muT_typing_context}
\end{figure}
\end{definition}

\begin{fact}
\label{fact:context_typing_correct}
Contextual typing judgments do indeed enjoy the intended behavior.
That is, we have $\mujudg \Gamma \Delta {\cctx E t} \sigma$ iff there is a type 
$\rho$ such that $\cjudg \Gamma \Delta E \rho \sigma$ and 
$\mujudg \Gamma \Delta t \rho$.
\end{fact}

\begin{fact}
\label{fact:muTv_typing_subst}
Typing is preserved under (structural) substitution.
\begin{enumerate}
\item If \mbox{$\mujudg {\Gamma, x : \rho} \Delta t \tau$} and 
	$\mujudg \Gamma \Delta r \rho$, then \mbox{$\mujudg \Gamma \Delta {\subst t x r} \tau$}.
\item If \mbox{$\mujudg \Gamma {\Delta, \alpha : \rho} t \tau$} and
	$\cjudg \Gamma \Delta E \rho \sigma$,
	then \mbox{$\mujudg \Gamma {\Delta, \beta : \sigma} {\subst t \alpha {\beta E}} \tau$}.
\end{enumerate}
We have corresponding results for commands.
\end{fact}

\begin{proof}
The first property is proven by mutual induction on the derivations of 
\mbox{$\mujudg {\Gamma, x : \rho} \Delta t \tau$} and
\mbox{$\musjudg {\Gamma, x : \rho} \Delta c$}. All cases are straightforward.
The second property is proven by induction on the derivations of
\mbox{$\mujudg \Gamma {\Delta, \alpha : \rho} t \tau$} and
$\musjudg \Gamma {\Delta, \alpha : \rho} c$. Most cases are straightforward,
so we only treat the passivate case. 
Let \mbox{$\musjudg \Gamma {\Delta, \alpha : \rho} {[\alpha]t}$}
with $\mujudg \Gamma {\Delta, \alpha : \rho} t \rho$. By the induction
hypothesis we have 
\mbox{$\mujudg \Gamma {\Delta, \beta : \sigma} {\subst t \alpha {\beta E}} \rho$}.
This leaves us to prove that 
\mbox{$\musjudg \Gamma {\Delta, \beta : \sigma} {\subst {([\alpha]t)} \alpha {\beta E}}$}.
Since 
$\subst {([\alpha]t)} \alpha {\beta E} \equiv [\beta]\cctx E {\subst t \alpha {\beta E}}$, 
the result follows from Fact~\ref{fact:context_typing_correct} and
the induction hypothesis.
\end{proof}

\begin{theorem}
\label{theorem:muT_subject_reduction}
\index{Subject reduction!for \lambdamuT{}}
The \lambdamucal{T} satisfies subject reduction.
\end{theorem}

\begin{proof}
We have to prove that all reduction rules preserve typing.
\begin{enumerate}
\item Proving that the result holds for the $\to_\beta$, $\to_0$ and 
	$\to_\sucT$-rule is straightforward, so we omit that.
\item To prove that the result holds for the $\to_{\mu R}$, $\to_{\mu\sucT}$ and 
	$\to_{\mu\natTypeT}$-rule it is sufficient to show that the result holds for
	$\cctx {E^s} {\mu \alpha.c} \to  \mu \beta.\subst c \alpha {\beta E^s}$
	by Fact~\ref{fact:muT_red_singular_context}. Given 
	$\mujudg \Gamma \Delta {\cctx {E^s} {\mu \alpha.c}} \tau$ we use 
	Fact~\ref{fact:context_typing_correct}	to obtain a type $\sigma$ such
	that $\mujudg \Gamma \Delta {\mu\alpha.c} \sigma$ and
	$\cjudg \Gamma \Delta E \sigma \tau$.
	\begin{small}\begin{flalign*}
	\AXC{$\musjudg \Gamma {\Delta, \alpha : \sigma} c$}
	\UIC{$\mujudg \Gamma \Delta {\mu\alpha.c} \sigma$}
	\AXC{$\cjudg \Gamma \Delta {E^s} \sigma \tau$}
	\BIC{$\mujudg \Gamma \Delta {\cctx {E^s} {\mu \alpha.c}} \tau$}
	\DisplayProof 
		& \ \to \ 
		\AXC{$\musjudg {\Gamma} {\Delta, \beta : \tau} {\subst c \alpha {\beta E^s}}$}
		\UIC{$\mujudg \Gamma \Delta {\mu\beta.\subst c \alpha {\beta E^s}} \tau$}
		\DisplayProof
	\end{flalign*}
	\end{small}
	\noindent Here we have
	$\musjudg {\Gamma} {\Delta, \beta : \tau} {\subst c \alpha {\beta E^s}}$
	by Fact~\ref{fact:muTv_typing_subst}.
\item For the $\to_{\mu\eta}$-rule we have the following.
	\begin{flalign*}
	\AXC{$\mujudg \Gamma {\Delta, \alpha: \rho} t \rho$}
	\UIC{$\musjudg \Gamma \Delta {[\alpha]t}$}
	\UIC{$\mujudg \Gamma \Delta {\mu \alpha.[\alpha]t} \rho$}
	\DisplayProof 
		& \ \to \ \mujudg \Gamma \Delta t \rho
	\end{flalign*}
	Since $\alpha\notin \FV t$, we have $\mujudg \Gamma \Delta t \rho$ by 
	strengthening.
\item For the $\to_{\mu i}$-rule we have the following.
	\begin{flalign*}
	\AXC{$\musjudg \Gamma {\Delta, \alpha:\rho, \beta:\rho} c$}
	\UIC{$\mujudg \Gamma {\Delta, \alpha:\rho} {\mu\beta.c} \rho$}
	\UIC{$\musjudg \Gamma {\Delta, \alpha:\rho} {[\alpha]\mu\beta.c}$}
	\DisplayProof 
		& \ \to \ 
		\AXC{$\musjudg \Gamma {\Delta, \alpha:\rho} {\subst c \beta {\alpha\ \Box}}$}
		\DisplayProof 
	\end{flalign*}
	Here we have $\musjudg \Gamma {\Delta, \alpha:\rho} {\subst c \beta {\alpha\ \Box}}$ 
	by Fact~\ref{fact:muTv_typing_subst} and the fact
	that $\cjudg \Gamma {\Delta, \alpha:\rho} {\Box} \rho \rho$. \qedhere
\end{enumerate}
\end{proof}

The $\to_\sucT$-rule, in contrast to the corresponding
rule of \lambdaT{} (Definition~\ref{definition:GodelsT_reduction}), 
only allows conversion when the numerical argument is a numeral. This 
restriction ensures that primitive recursion is not performed on terms that 
might reduce to a term of the shape $\mu\alpha.c$. If we omit this restriction 
we lose confluence. 
\begin{example}
We illustrate this by considering a variant of our system 
with the following rule instead.
\begin{flalign*}
	\nrecT\ r\ s\ (\sucT t) \to{} & s\ t\ (\nrecT\ r\ s\ t) \tag{$\sucT'$}
\end{flalign*}
Now we can reduce the term $t \equiv \mu \alpha.[\alpha]\nrecT\ 0\ (\lm x h.\natenc 2)\ (\sucT\muthrow[\alpha]\natenc 4)$
to two distinct normal forms:
\begin{flalign*}
t \equiv{}& \mu \alpha.[\alpha]\nrecT\ 0\ (\lm x h.\natenc 2)\ (\sucT\muthrow[\alpha]\natenc 4)\\
	\to{} & \mu \alpha.[\alpha]\nrecT\ 0\ (\lm x h.\natenc 2)\ (\muthrow[\alpha]\natenc 4) \tag{$\mu\sucT$} \\
	\to{} & \mu \alpha.[\alpha]\muthrow[\alpha]\natenc 4 \tag{$\mu\natTypeT$}\\
	\to{} & \mu \alpha.[\alpha]\natenc 4 \tag{$\mu i$} \\
	\to{} & \natenc 4 \tag{$\mu\eta$}
\end{flalign*}
and:
\begin{flalign*}
t \equiv{}& \mu \alpha.[\alpha]\nrecT\ 0\ (\lm x h.\natenc 2)\ (\sucT\muthrow[\alpha]\natenc 4) \\
	\to{} & \mu \alpha.[\alpha](\lm x h.\natenc 2)\ (\muthrow[\alpha]\natenc 4)\ (\nrecT\ 0\ (\lm x h.\natenc 2)\ (\muthrow[\alpha]\natenc 4)) \tag{$\sucT'$} \\
	\tto{} & \mu \alpha.[\alpha]\natenc 2 \tag{$\beta$} \\
	\to{} & \natenc 2 \tag{$\mu\eta$}
\end{flalign*}
\end{example}

Alternatively, in order to obtain a confluent system, it is possible to remove
the $\to_{\sucT}$-rule while retaining the unrestricted $\to_{\mu\sucT'}$-rule. However, 
then we can construct closed terms $t \typeof \natTypeT$ that are in normal form 
but are not a numeral.
An example of such a term is $\mu\alpha.[\alpha]\sucT\mu\beta.[\alpha]0$. 

\begin{lemma}
\label{muT:typed_values}
Given a value $v$ such that $\mujudg {} \Delta v \rho$, we have:
\begin{enumerate}
\item If $\rho = \natTypeT$, then $v \equiv \natenc n$.
\item If $\rho = \sigma \to \tau$, then $t \equiv \lm x.r$ for some variable $x$ and term $r$. 
\end{enumerate}
\end{lemma}

\begin{proof}
This result is proven by induction on the structure of values.
\end{proof}

\begin{lemma}
\label{lemma:muT_nf_open}
Given a term $t$ that is in normal and such that $\mujudg {} \Delta t \rho$, then
$t$ is a value or $t \equiv \mu\alpha.[\beta]v$ for some value $v$.
\end{lemma}

\begin{proof}
By induction on the derivation $\mujudg {} {\Delta} t \rho$.
\begin{enumerate}
\item[(var)] Let $\mujudg {} \Delta x \rho$ with $x : \rho \in \emptyset$. Now we obtain a contradiction since $x : \rho \notin \emptyset$.
\item[($\lambda$)] Let $\mujudg {} \Delta  {\lambda x.r} {\sigma \to \tau}$. Now we are immediately done.
\item[(app)] Let $\mujudg {} \Delta {rs} \tau$ with $\mujudg {} \Delta  r {\sigma \to \tau}$ and 
	$\mujudg {} \Delta  s \sigma$. Now by the induction hypothesis and Lemma~\ref{muT:typed_values} we have $r \equiv \lm x.r'$
	or $r \equiv \mu \alpha.[\beta]v$. But since $rs$ should be in normal form we 
	obtain a contradiction.
\item[(zero)] Let $\mujudg {} \Delta  0 \natTypeT$. Now we are immediately done.
\item[(suc)] Let $\mujudg {} \Delta  {\sucT t} \natTypeT$ with $\mujudg {} \Delta  t \natTypeT$. 
	Now we have $t \equiv \natenc n$ or	$t \equiv \mu\alpha.[\beta]v$ 
	by the induction hypothesis and Lemma~\ref{muT:typed_values}. In the former case
	we are immediately done, in the latter case we obtain a contradiction because
	the $\to_{\mu \sucT}$-rule can be applied.
\item[(nrec)] Let $\mujudg {} \Delta  {\nrecT\ r\ s\ t} \rho$ with $\mujudg {} \Delta  t \natTypeT$. 
	Now we have $t \equiv \natenc n$ or	$t \equiv \mu\alpha.[\beta]v$ by the 
	induction hypothesis and Lemma~\ref{muT:typed_values}.
	But in both cases we obtain a contradiction because the reduction rules $\to_{\mu 0}$, $\to_{\mu \sucT}$ and 
	$\to_{\mu \natTypeT}$ can be applied, respectively.
\item[(act/pas)] Let $\mujudg {} \Delta {\mu \alpha.[\beta]t} \rho$ with
	$\mujudg {} {\Delta, \alpha: \rho} t \tau$ and $\beta : \tau \in (\Delta, \alpha: \rho)$. 
	Now we have that $t$ is a value or $t \equiv \mu\alpha.[\beta]v$ by the induction hypothesis.
	In the former case we are immediately done, in the latter case we obtain a contradiction
	because the $\to_{\mu i}$-rule can be applied. \qedhere
\end{enumerate}
\end{proof}

\begin{theorem}
\label{theorem:muT_nf}
Given a term $t$ that is in normal form and such that \mbox{$\mujudg {} {} t \natTypeT$}, then 
$t \equiv \natenc n$ for some $n \in \nat$.
\end{theorem}

\begin{proof}
By Lemma~\ref{lemma:muT_nf_open} we obtain that $t \equiv v$ or 
$t \equiv \mu\alpha.[\beta]v$ for some value $v$. In the former case we have
$t \equiv \natenc n$ by Lemma~\ref{muT:typed_values}.
In the latter case we have $\beta = \alpha$ since $t$ is closed for 
$\mu$-variables, so $t \equiv \mu\alpha.[\alpha]\natenc n$ by Lemma~\ref{muT:typed_values}.
But now we obtain a contradiction because we can apply the $\to_{\mu \eta}$-rule.
\end{proof}

%% file: arithmetic_cps.tex
\section{CPS-translation of \texorpdfstring{$\lambdamuT$}{lambda-mu-T} into \texorpdfstring{$\lambdaT$}{lambda-T}}
\label{section:muT_cps}
\index{CPS-translation!of $\lambdamuT$ into $\lambdaT$}
In this section we will present a CPS-translation from \lambdamuT{} 
into \lambdaT. We will use this CPS-translation to prove the
main result of this section: the functions that are representable in \lambdamuT{} are exactly
the functions that are provably recursive in first-order arithmetic.

\begin{definition}
\label{definition:muT_negative_translation}
\index{Negative translation}
Let $\neg \rho$ denote $\rho \to \bot$ for a fixed type $\bot$. Given
a type $\rho$, the \emph{negative translation} $\muTCPS \rho$ of $\rho$ 
is mutually inductively defined with $\muTCPSt \rho$ as follows.
\begin{flalign*}
	\muTCPS \rho & \defined \neg\neg\muTCPSt \rho\\ 
	\muTCPSt \natTypeT &\defined \natTypeT \\
	\muTCPSt {(\sigma \to \tau)} &\defined \muTCPS\sigma \to \muTCPS\tau
\end{flalign*}
\end{definition}

\begin{definition}
Given $\lambdaT$-terms $t$ and $r$, \emph{the CPS-application} 
$\muTCPSapp t r$ of $t$ and $r$ is defined as follows.
\[
	\muTCPSapp t r \defined \lm k.t(\lm l.lrk)
\]
\end{definition}

\begin{definition}
Given a $\lambdaT$-term $t$, \emph{the negative} of $t$ is defined as 
follows.
\[
	\natneg t \defined \lm k.kt
\]
\end{definition}

\begin{fact}
\label{fact:muT_cps_app_welltyped}
If $\ljudg \Gamma t {\muTCPS{(\sigma \to \tau)}}$ and 
\mbox{$\ljudg \Gamma r {\muTCPS\sigma}$}, 
then $\ljudg \Gamma {\muTCPSapp t r} {\muTCPS\tau}$.
\end{fact}

\begin{definition}
\label{definition:muT_cps}
Given a $\lambdamuT$-term $t$, then the \emph{CPS-translation} $\muTCPS t$ of $t$
into $\lambdaT$ is inductively defined as follows.
\begin{flalign*}
	\muTCPS x &\defined \lm k.xk \\
	\muTCPS{(\lm x.t)} &\defined \lm k.k (\lm x.\muTCPS t) \\
	\muTCPS{(tr)} &\defined \muTCPSapp {\muTCPS t} {\muTCPS r} \\
	\muTCPS 0 &\defined \natneg 0 \\
	\muTCPS{(\sucT t)} &\defined \lm k.\muTCPS t (\lm l.k(\sucT l)) \\
	\muTCPS{(\nrecT_\rho\ r\ s\ t)} &\defined 
		\lm k.\muTCPS t(\lm l.\nrecT\ \muTCPS r\ s'\ l\ k) \\
		& {\text where }\ s' \defined \lm xp.\muTCPSapp {\muTCPSapp {\muTCPS s} {\natneg x}} p \\
	\muTCPS{(\mu\alpha.c)} &\defined \lm k_\alpha . \muTCPS c \\
	\muTCPS{([\alpha]t)} &\defined \muTCPS t k_\alpha
\end{flalign*}
Here $k_\alpha$ is a fresh $\lambda$-variable for each $\mu$-variable $\alpha$.
\end{definition}

In the translation of $\nrecT_\rho\ r\ s\ t$ we see that we are required to evaluate
$t$ first, simply because it is the only way to obtain a numeral from $t$.

\begin{fact}
\label{fact:muT_cps_natneg_welltyped}
If $\ljudg \Gamma t \natTypeT$, then 
$\ljudg \Gamma {\natneg t} {\muTCPS\natTypeT}$.
\end{fact}

\begin{theorem}
\label{theorem:muT_cps_welltyped}
The translation from $\lambdamuT$ into $\lambdaT$ preserves typing.
That is:
\[
	\mujudg \Gamma \Delta t \rho\text{ in }\lambdamuT 
		\qquad \implies \qquad
	\ljudg {\muTCPS\Gamma,\muTCPS\Delta} {\muTCPS t} {\muTCPS\rho}\text{ in }\lambdaT
\]
where $\muTCPS\Gamma = \{ x : \muTCPS\rho \separator x : \rho \in \Gamma \}$
and $\muTCPS\Delta = \{ k_\alpha : \neg\muTCPSt\rho \separator \alpha : \rho \in \Delta \}$.
\end{theorem}

\begin{proof}
We prove that we have $\mujudg \Gamma \Delta t \rho$ and $\musjudg \Gamma \Delta c$
by mutual induction on the derivations $\ljudg {\muTCPS\Gamma,\muTCPS\Delta} {\muTCPS t} {\muTCPS\rho}$ and
$\ljudg {\muTCPS\Gamma,\muTCPS\Delta} {\muTCPS t} \bot$, respectively. Most of the
cases are straightforward, so we treat just one interesting case.
\begin{enumerate}
\item[(nrec)] Let $\mujudg \Gamma \Delta {\nrecT_\rho\ r\ s\ t} \rho$ with
	$\mujudg \Gamma \Delta r \rho$, $\mujudg \Gamma \Delta s {\natTypeT\to\rho\to\rho}$ 
	and $\mujudg \Gamma \Delta t \natTypeT$. Now we
	have $\ljudg {\muTCPS\Gamma,\muTCPS\Delta} {\muTCPS r} {\muTCPS\rho}$,
	$\ljudg {\muTCPS\Gamma,\muTCPS\Delta} {\muTCPS s} {\muTCPS{(\natTypeT \to \rho \to \rho)}}$
	and $\ljudg {\muTCPS\Gamma,\muTCPS\Delta} {\muTCPS t} {\muTCPS\natTypeT}$
	by the induction hypothesis. Furthermore we have 
	$s' \equiv \lm xp.\muTCPSapp {\muTCPSapp {\muTCPS s} {\natneg x}} p \typeof \natTypeT \to \muTCPS \rho \to \muTCPS \rho$
	as shown below.
	\begin{prooftree}
	\AXC{$\muTCPS s \typeof \muTCPS{(\natTypeT \to \rho \to \rho)}$}
		\AXC{$x \typeof \natTypeT$}
		\RightLabel{(a)}
		\UIC{$\natneg x \typeof \muTCPS\natTypeT$}
	\RightLabel{(b)}
	\BIC{$\muTCPSapp {\muTCPS s} {\natneg x} \typeof \muTCPS{(\rho \to \rho)}$}
		\AXC{$p \typeof \muTCPS \rho$}
	\RightLabel{(c)}
	\BIC{$\muTCPSapp {\muTCPSapp {\muTCPS s} {\natneg x}} p \typeof \muTCPS \rho$}
	\UIC{$\lm xp.\muTCPSapp {\muTCPSapp {\muTCPS s} {\natneg x}} p \typeof \natTypeT \to \muTCPS \rho \to \muTCPS \rho$}
	\end{prooftree}
	Here, step (a) follows from Fact~\ref{fact:muT_cps_natneg_welltyped} and step (b)
	and (c) follow from  Fact~\ref{fact:muT_cps_app_welltyped}.
	So	\mbox{$\ljudg {\muTCPS\Gamma,\muTCPS\Delta} {\muTCPS{(\nrecT_\rho\ r\ s\ t)}} {\muTCPS\rho}$}
	as shown below.
	\[
	\AXC{$\muTCPS t \typeof \muTCPS\natTypeT$}
			\AXC{$\muTCPS r \typeof \muTCPS \rho$}
			\AXC{$s' \typeof \natTypeT \to \muTCPS \rho \to \muTCPS \rho$}
			\AXC{$l \typeof \natTypeT$}
		\TIC{$\nrecT\ \muTCPS r\ s'\ l \typeof \muTCPS \rho$}
			\AXC{$k \typeof \neg\muTCPSt \rho$}
		\BIC{$\nrecT\ \muTCPS r\ s'\ l\ k \typeof \bot$}
		\UIC{$\lm l.\nrecT\ \muTCPS r\ s'\ l\ k \typeof \neg\natTypeT$}
	\BIC{$\muTCPS t(\lm l.\nrecT\ \muTCPS r\ s'\ l\ k) \typeof \bot$}
	\UIC{$\lm k.\muTCPS t(\lm l.\nrecT\ \muTCPS r\ s'\ l\ k) \typeof \muTCPS\rho$}
	\normalAlignProof
	\DisplayProof
	\]
	\qedafterarray
\end{enumerate}
\end{proof}

\begin{fact}
\label{fact:muT_cps_nat}
For each $n \in \nat$ we have $\muTCPS{\natenc n} \tto \natneg{\natenc n}$.
\end{fact}

\begin{proof}
By induction on $n$. 
\begin{enumerate}
\item Let $n = 0$. We have $\muTCPS{\natenc 0} \equiv \natneg{\natenc 0}$
	by Definition~\ref{definition:muT_cps}.
\item Let $n > 0$. We have $\muTCPS{\natenc n} \tto \natneg{\natenc n}$
	by the induction hypothesis and hence:
	\begin{flalign*}
	\muTCPS{\natenc{n+1}} 
		\equiv{}& \lm k.\muTCPS{\natenc n}(\lm l.k(\sucT l)) \\
		\tto{}& \lm k.(\lm q.q\natenc n)(\lm l.k(\sucT l)) \\
		\tto{}& \lm k.k(\sucT \natenc n)\\
		\equiv{}& \natneg{\natenc{n+1}}
	\end{flalign*}
	\qedafterarray
\end{enumerate}
\end{proof}

\begin{lemma}
\label{lemma:muT_cps_abs}
For each term $t$ we have $\lambda k.\muTCPS t k \to \muTCPS t$.
\end{lemma}

\begin{proof}
This follows immediately from the Definition~\ref{definition:muT_cps} since 
the translation $\muTCPS t$ of a term $t$ is of the shape $\lm l.t'$, 
so $\lm k.(\lm l.t')k \tto \lm k.\subst {t'} l k \equiv \muTCPS t$.
\end{proof}

\begin{lemma}
\label{lemma:muT_cps_nrec_abs}
We have $\lm k.\nrecT\ \muTCPS r\ s'\ \natenc n\ k = \nrecT\ \muTCPS r\ s'\ {\natenc n}$
for $s' \equiv \lm xp.\muTCPSapp {\muTCPSapp {\muTCPS s} {\natneg x}} p$.
\end{lemma}

\begin{proof}
We distinguish the following cases.
\begin{enumerate}
\item Let $n = 0$. The result follows from Lemma~\ref{lemma:muT_cps_abs}.
\item Let $n > 0$. Now we have the following.
	\begin{flalign*}
	\lm k.\nrecT\ \muTCPS r\ s'\ \natenc n\ k
		\tto{}& \lm k.s'\ \natenc{n-1}\ (\nrecT\ \muTCPS r\ s'\ \natenc{n-1})\ k \\
		\tto{}& \lm k.(\muTCPSapp {\muTCPSapp {\muTCPS s} {\natneg {\natenc{n-1}}}} {\nrecT\ \muTCPS r\ s'\ \natenc{n-1}})\ k \\
		\equiv{}& \lm k.(\lm k_2.(\muTCPSapp {\muTCPS s} {\natneg {\natenc{n-1}}})\ (\lm l.l\ (\nrecT\ \muTCPS r\ s'\ \natenc{n-1})\ k_2))\ k \\
		\tto{}& \lm k.(\muTCPSapp {\muTCPS s} {\natneg {\natenc{n-1}}})\ (\lm l.l\ (\nrecT\ \muTCPS r\ s'\ \natenc{n-1})\ k) \\
		\equiv{}& \muTCPSapp {\muTCPSapp {\muTCPS s} {\natneg {\natenc{n-1}}}} {\nrecT\ \muTCPS r\ s'\ \natenc{n-1}} \\
		={}& s'\ \natenc{n-1}\ (\nrecT\ \muTCPS r\ s'\ \natenc{n-1})\\
		={}& \nrecT\ \muTCPS r\ s'\ \natenc n
	\end{flalign*}
	\qedafterarray
\end{enumerate}
\end{proof}

\begin{lemma}
\label{lemma:muT_cps_subst}
The translation from $\lambdamuT$ into $\lambdaT$ preserves (structural) 
substitution. That is:
\begin{enumerate}
\item $\subst {\muTCPS t} x {\muTCPS r} \tto \muTCPS{(\subst t x r)}$
\item $\muTCPS{(\subst t \alpha {\beta\ \Box})} \equiv \subst {\muTCPS t} {k_\alpha} {k_\beta}$
\item $\muTCPS{(\subst t \alpha {\beta\ (\sucT \Box)})} \tto \subst {\muTCPS t} {k_\alpha} {\lm l.k_\beta(\sucT l)}$
\item $\muTCPS{(\subst t \alpha {\beta\ (\Box s)})} \tto \subst {\muTCPS t} {k_\alpha} {\lm l.l \muTCPS s k_\beta}$
\item $\muTCPS{(\subst t \alpha {\beta\ (\nrecT\ r\ s\ \Box)})} \tto \subst {\muTCPS t} {k_\alpha} {\lm l.\nrecT\ \muTCPS r\ s'\ l\ k_\beta}$
\end{enumerate}
\end{lemma}

\begin{proof}
These results are proven by induction on the structure of $t$.
\end{proof}

\begin{lemma}
\label{lemma:muT_cps_sound}
The translation from $\lambdamuT$ into $\lambdaT$ preserves convertibility.
That is, if $t_1 = t_2$, then $\muTCPS{t_1} = \muTCPS{t_2}$.
\end{lemma}

\begin{proof}
By induction on the derivation of $t_1 \to t_2$. Most of the
cases are straightforward, so we treat just one interesting case.
\begin{enumerate}
\item Let $\nrecT\ r\ s\ (\sucT \natenc n) \to s\ \natenc n\ (\nrecT\ r\ s\ \natenc n)$. Now:
	\begin{flalign*}
	\muTCPS{(\nrecT\ r\ s\ (\sucT \natenc n))} 
		\equiv{}& \lm k.\muTCPS{(\sucT \natenc n)}\, (\lm l.\nrecT\ \muTCPS r\ s'\ l\ k) \\
		\tto{}& \lm k.\natneg{\sucT \natenc n}\ (\lm l.\nrecT\ \muTCPS r\ s'\ l\ k) \tag{a}\\
		\tto{}& \lm k.\nrecT\ \muTCPS r\ s'\ (\sucT \natenc n)\ k \\
		\to{}& \lm k.s'\ \natenc n\ (\nrecT\ \muTCPS r\ s'\ \natenc n)\ k\\
		\tto{}& \lm k.(\muTCPSapp {\muTCPSapp {\muTCPS s} {\natneg{\natenc n}}} {\nrecT\ \muTCPS r\ s'\ \natenc n})\ k\\ \displaybreak[1]
		={}& \lm k.(\lm k_2.({\muTCPSapp {\muTCPS s} {\natneg{\natenc n}}})\
			(\lm l.l\ (\nrecT\ \muTCPS r\ s'\ \natenc n)\; k_2))\ k \\
		={}& \lm k.(\muTCPSapp {\muTCPS s} {\natneg {\natenc n}})\ 
			(\lm l.l\ (\nrecT\ \muTCPS r\ 's\ {\natenc n})\ k)\\
		={}& \muTCPSapp {\muTCPSapp {\muTCPS s} {\natneg {\natenc n}}} {\nrecT\ \muTCPS r\ s'\ {\natenc n}}\\
		={}& \muTCPSapp {\muTCPSapp {\muTCPS s} {\natneg {\natenc n}}} {\lm k_2.\nrecT\ \muTCPS r\ s'\ {\natenc n}\ k_2} \tag{b}\\
		={}& \muTCPSapp {\muTCPSapp {\muTCPS s} {\natneg {\natenc n}}} {\lm k_2.\natneg{\natenc n}\ (\lm l.\nrecT\ \muTCPS r\ s'\ l\ k_2)}\\
		={}& \muTCPSapp {\muTCPSapp {\muTCPS s} {\natneg {\natenc n}}} {\lm k_2.\muTCPS{\natenc n}\, (\lm l.\nrecT\ \muTCPS r\ s'\ l\ k_2)} \tag{c} \\
		\equiv{}& \muTCPS{(s\ \natenc n\ (\nrecT\ r\ s\ \natenc n))}
	\end{flalign*}
	Here, step (a) holds by Fact~\ref{fact:muT_cps_nat}, step (b) holds
	by Lemma~\ref{lemma:muT_cps_nrec_abs} and step (c) holds by Fact~\ref{fact:muT_cps_nat}. \qedhere
\end{enumerate}
\end{proof}

\begin{theorem}
\label{theorem:muT_representable}
Each function $f : \nat^n \to \nat$ that is representable
in $\lambdamuT$ is representable in $\lambdaT$. That is, if a
term $t$ with $\mujudg {} {} t {\natTypeT^n \to \natTypeT}$ represents the function 
$f$ in $\lambdamuT$, then there exists a term $t'$ with 
$\ljudg {} {t'} {\natTypeT^n \to \natTypeT}$ that represents 
the function $f$ in $\lambdaT$.
\end{theorem}

\begin{proof}
Suppose that $t : \natTypeT^n \to \natTypeT$ represents $f : \nat^n \to \nat$ in $\lambdamuT$.
That means that $\natenc{f(m_1, \ldots, m_n)} = t\; \natenc {m_1} \ldots \natenc {m_n}$.
Now define a term $t'$ as follows.
\[
	t' \defined \lm x_1 : \natTypeT \ldots \lm x_n : \natTypeT\ .\ (
		\muTCPSapp {\muTCPSapp {\muTCPSapp {\muTCPS t} {\natneg{x_1}}} \ldots} {\natneg{x_n}})\ 
		(\lm x : \natTypeT \ .\ x)
\]
Now we have $\muTCPS t : \muTCPS{(\natTypeT^n \to \natTypeT)}$ by Theorem~\ref{theorem:muT_cps_welltyped},
$x_i : \muTCPS{\natTypeT}$ by Fact~\ref{fact:muT_cps_natneg_welltyped} and 
therefore $\muTCPSapp {\muTCPSapp {\muTCPSapp {\muTCPS t} {\natneg{x_1}}} \ldots} {\natneg{x_n}} : \muTCPS{\natTypeT}$
by Fact~\ref{fact:muT_cps_app_welltyped}. Hence by setting $\bot = \natTypeT$ we
have $t' : \natTypeT$. Now it remains to prove that $\natenc{f(m_1, \ldots, m_n)} = t'\; \natenc {m_1} \ldots \natenc {m_n}$.
\begin{flalign*}
	t'\; \natenc {m_1} \ldots \natenc {m_n} 
		&= (\muTCPSapp {\muTCPSapp {\muTCPSapp {\muTCPS t} {\natneg{\natenc{m_1}}}} \ldots} {\natneg{\natenc{m_n}}})\ \lm x.x \\
		&= (\muTCPSapp {\muTCPSapp {\muTCPSapp {\muTCPS t} {\muTCPS{\natenc{m_1}}}} \ldots} {\muTCPS{\natenc{m_n}}})\ \lm x.x \tag{a}\\
		&= \muTCPS{(t\; \natenc{m_1} \ldots \natenc{m_n})}\ \lm x.x \\
		&= \muTCPS{(\natenc{f(m_1, \ldots, m_n)})}\ \lm x.x \tag{b}\\
		&= \natneg{\natenc{f(m_1, \ldots, m_n)}}\ \lm x.x \tag{c}\\
		&= \natenc{f(m_1, \ldots, m_n)} 
\end{flalign*}
Here, step (a) holds by Fact~\ref{fact:muT_cps_nat}, step (b) holds
by Lemma~\ref{lemma:muT_cps_sound} and step (c) holds by Fact~\ref{fact:muT_cps_nat}.
\end{proof}

\begin{corollary}
The functions representable in $\lambdamuT$ are exactly those that are
provably recursive in first-order arithmetic.
\end{corollary}

\begin{proof}
This result follows immediately from Theorem~\ref{theorem:godelsT_representable} 
and~\ref{theorem:muT_representable}.
\end{proof}

%% file: arithmetic_confluence.tex
\section{Confluence of \texorpdfstring{\lambdamuT}{lambda-mu-T}}
\label{section:muT_confluence}
\index{Confluence!for \lambdamuT}
\index{Parallel reduction}
\index{Complete development}
To prove confluence one typically uses the notion of \emph{parallel
  reduction}, as introduced by Tait and Martin-L\"of. Intuitively, a
parallel reduction relation $\Rightarrow$ allows to contract a number
of redexes in a term simultaneously. Following Takahashi~\cite{takahashi1995},
$\Rightarrow$ can be defined by induction over the term
structure, making it easy to prove that it is preserved under
substitution. Then one proves that $\Rightarrow$ satisfies:
\begin{itemize}
\item The \emph{diamond property}: if $t_1 \Rightarrow t_2$ and 
	$t_1\Rightarrow t_3$, then there exists a $t_4$	such that $t_2 \Rightarrow t_4$ 
	and $t_3 \Rightarrow t_4$, in a diagram:
\[\xymatrix{
                & t_1 \ar@{=>}[ld] \ar@{=>}[rd] & \\
t_2 \ar@{:>}[rd]&                               & t_3 \ar@{:>}[ld] \\
                & t_4                           &
}\]
\item $\mathrm{\Rightarrow} \subset \mathrm{\tto}$: if $t_1 \Rightarrow t_2$, then $t_1 \tto t_2$.
\item $\mathrm{\tto} \subset \mathrm{\Rightarrow^*}$: if $t_1 \tto t_2$, then $t_1 \Rightarrow^* t_2$.
\end{itemize}
Thus one obtains confluence of $\to$.
To streamline proving the diamond property of $\Rightarrow$ one can
define the \emph{complete development} of a term $t$, notation $\MPRED t$, 
which is obtained by contracting all redexes in $t$. Now it suffices
to prove that $t_1 \Rightarrow t_2$ implies $t_2 \Rightarrow \MPRED {t_1}$. 
Unfortunately, as observed in~\cite{fujita1997, baba2001},
adopting the notion of parallel reduction in a standard way
does not work for \lambdamu{}. The resulting parallel reduction
relation will only be weakly confluent and not confluent.

In this section we will focus on resolving this problem for
\lambdamuT{}.  For an extensive discussion of parallel reduction and
its application to various systems we refer to~\cite{takahashi1995}.
A simple-minded parallel reduction relation, 
obtained by extending Parigot's parallel reduction~\cite{parigot1992} 
to \lambdamuT{}, would have the follow rules:
\begin{enumerate}
\item[\em (t6.1)] If $c \Rightarrow c'$, then $\mu \alpha.c \Rightarrow \mu \alpha.c'$.
\item[\em (t6.2)] If $c \Rightarrow c'$ and $s \Rightarrow s'$,
	then $(\mu \alpha.c)s \Rightarrow \mu\alpha.\subst {c'} \alpha {\alpha\ (\Box s')}$.
\item[\em (t6.3)] If $c \Rightarrow c'$, 
	then $\sucT (\mu\alpha.c) \Rightarrow \mu\alpha.\subst {c'} \alpha {\alpha\ (\sucT\Box)}$.
\item[\em (t6.4)] If $r \Rightarrow r'$, $s \Rightarrow s'$ and $c \Rightarrow c'$,
	then \\ $\nrecT\ r\ s\ \mu\alpha.c \Rightarrow \mu\alpha.\subst {c'} \alpha {\alpha\ (\nrecT\ r'\ s'\ \Box)}$.
\item[\em (t7)] If $t \Rightarrow t'$ and $\alpha\notin \FCV t$, 
	then $\mu \alpha.[\alpha]t \Rightarrow t'$.
\item[\em (c1)] If $t \Rightarrow t'$, then $[\alpha]t \Rightarrow [\alpha]t'$.
\item[\em (c2)] If $c \Rightarrow c'$, then $[\alpha]\mu\beta.c \Rightarrow \subst {c'} \beta {\alpha \Box}$.
\end{enumerate}
As has been observed in~\cite{fujita1999}, Parigot's original parallel
reduction relation is not confluent. Similarly, the
parallel reduction as defined above for \lambdamuT{} is not confluent.
Let us (as in~\cite{baba2001}) consider the term $(\mu
\alpha.[\alpha]\mu\gamma.[\alpha]x)y$, this term contains both a
(t6.2) and a (c2)-redex. However, after contracting the (t6.2)-redex,
we obtain the term $\mu \alpha.[\alpha](\mu\gamma.[\alpha]xy)y$, in
which the (c2)-redex is blocked.
\[\xymatrix{
& (\mu \alpha.[\alpha]\mu\gamma.[\alpha]x)y \ar@{=>}[ld] \ar@{=>}[rd]
  & \\ \mu \alpha.[\alpha](\mu\gamma.[\alpha]xy)y \ar@{:>}[d] & & (\mu
  \alpha.[\alpha]x)y \ar@{:>}[d] \\ \mu
  \alpha.[\alpha]\mu\gamma.[\alpha]xy \ar@{:>}[rr] & & \mu
  \alpha.[\alpha]xy }\] Although it is possible to prove that this
relation is weakly confluent, weak confluence is not quite
satisfactory. Of course, since \lambdamuT{} is strongly normalizing
(Theorem~\ref{theorem:muT_sn}), it would give confluence by Newman's
lemma. However, an untyped version of \lambdamuT{} is of course not
strongly normalizing, hence we do not obtain confluence for raw terms
this way.

Baba, Hirokawa and Fujita~\cite{baba2001} noticed that this problem
could be repaired by allowing a $\mu \beta$ to ``jump over a whole
context'' to its corresponding $[\alpha]$. Their version of the (c2)-rule 
is as follows.
\begin{enumerate}
\item[\em (c2)] If $c \Rightarrow c'$ and $E \Rightarrow E'$, then $[\alpha]\cctx E {\mu\beta.c} \Rightarrow \subst {c'} \beta {\alpha E'}$.
\end{enumerate}
Here $E$ and $E'$ are contexts and parallel reduction on contexts is
defined by reducing all its components in parallel. This (c2)-rule
performs ``deep'' structural substitutions and renaming in one step
and thus covers and extends the original rules (t6.1-4) and (c2)

Baba \etal{}~\cite{baba2001} have shown that their relation
$\Rightarrow$ is confluent for \lambdamu{} without the (t7) rule. It
is not confluent if the (t7) rule is included. Let us (as in
\cite{baba2001}) consider the term $\mu
\alpha.[\alpha](\mu\beta.[\gamma]x)yz$.
\[\xymatrix{
& \mu \alpha.[\alpha](\mu\beta.[\gamma]x)yz \ar@{=>}[ldd] \ar@{=>}[rd]
  & \\ & & (\mu \beta.[\gamma]x)yz \ar@{:>}[d] \\ \mu \alpha.[\gamma]x
  & & (\mu \beta.[\gamma]x)z \ar@{:>}[ll] }\] In the conclusion of
their work they suggest that this problem can be repaired by
considering a series of structural substitutions (t6.1-4) as one
step. This approach has been carried out successfully by Nakazawa for
a \cbv{} variant of \lambdamu{}~\cite{nakazawa2003}. However,
Nakazawa did not use the notion of complete development. We will
follow the approach suggested by Baba \etal{} for \lambdamuT{} and
use the notion of complete development.

\begin{definition}
\label{definition:muT_pred}
\emph{Parallel reduction $t \Rightarrow t'$ on terms} 
is mutually inductively defined with \emph{parallel reduction $c \Rightarrow c'$ 
on commands} 
and \emph{parallel reduction $E \Rightarrow E'$ on contexts} 
as follows.
\begin{enumerate}
\item[\em (t1)] $x \Rightarrow x$
\item[\em (t2)] $0 \Rightarrow 0$
\item[\em (t3)] If $t \Rightarrow t'$, then $\lm x.t \Rightarrow \lm x.t'$.
\item[\em (t4)] If $t \Rightarrow t'$ and $E^s \Rightarrow E'$, then $\cctx {E^s} t \Rightarrow \cctx {E'} {t'}$.
\item[\em (t5)] If $t \Rightarrow t'$ and $r \Rightarrow r'$, then $(\lm x.t)r \Rightarrow \subst {t'} x {r'}$.
\item[\em (t6)] If $c \Rightarrow c'$ and $E \Rightarrow E'$, then $\cctx E {\mu \alpha.c} \Rightarrow \mu\alpha.\subst {c'} \alpha {\alpha E'}$.
\item[\em (t7)] If $t \Rightarrow t'$ and $\alpha\notin \FCV t$, then $\mu \alpha.[\alpha]t \Rightarrow t'$.
\item[\em (t8)] If $r \Rightarrow r'$, then $\nrecT\ r\ s\ 0 \Rightarrow r'$.
\item[\em (t9)] If $r \Rightarrow r'$ and $s \Rightarrow s'$, then $\nrecT\ r\ s\ (\sucT \natenc n) \Rightarrow s'\ \natenc n\ (\nrecT\ r'\ s'\ \natenc n) $.
\end{enumerate}

\begin{enumerate}
\item[\em (c1)] If $t \Rightarrow t'$, then $[\alpha]t \Rightarrow [\alpha]t'$.
\item[\em (c2)] If $c \Rightarrow c'$ and $E \Rightarrow E'$, then $[\alpha]\cctx E {\mu\beta.c} \Rightarrow \subst {c'} \beta {\alpha E'}$.
\end{enumerate}

\begin{enumerate}
\item[\em (E1)] $\Box \Rightarrow \Box$
\item[\em (E2)] If $E \Rightarrow E'$ and $t\Rightarrow t'$, then $Et \Rightarrow E't'$.
\item[\em (E3)] If $E \Rightarrow E'$, then $\sucT E \Rightarrow \sucT E'$.
\item[\em (E4)] If $E \Rightarrow E'$, $r\Rightarrow r'$ and $s\Rightarrow s'$, then $\nrecT\ r\ s\ E \Rightarrow \nrecT\ r'\ s'\ E'$.
\end{enumerate}
Furthermore, $\Rightarrow^*$ denotes the transitive closure of $\Rightarrow$.
\end{definition}

For conciseness of presentation, we specify most of the forthcoming
lemmas just for terms. Yet they can always be mutually
stated and mutually inductively proven for commands and contexts.

\begin{lemma}
\label{lemma:muT_pred_refl}
Parallel reduction is reflexive. That is, $t \Rightarrow t$ for all terms $t$.
\end{lemma}

\begin{proof}
By induction on $t$. We use the rules (t1-4), (t6), (c1) and (E1-4).
\end{proof}

\begin{lemma}
\label{lemma:muT_pred_cctx_subst}
If $E \Rightarrow E'$ and $t \Rightarrow t'$, then $\cctx {E} t \Rightarrow \cctx {E'} {t'}$.
\end{lemma}

\begin{proof}
By induction on the derivation of $E \Rightarrow E'$.
\end{proof}

\begin{lemma}
If $E^s$ is singular and $E^s \Rightarrow E'$, then $E'$ is singular.
\end{lemma}

\begin{proof}
By a case analysis on the derivation of $E^s \Rightarrow E'$.
\end{proof}

\begin{lemma}
\label{lemma:muT_pred_fv}
If $t \Rightarrow t'$, then $\FV{t'} \subseteq \FV t$ and 
$\FCV{t'} \subseteq \FCV t$.
\end{lemma}

\begin{proof}
By induction on the derivation of $t \Rightarrow t'$.
\end{proof}

\begin{lemma}
\label{lemma:muT_pred_subst}
\label{lemma:muT_pred_strucsubst}
Parallel reduction is preserved under (structural) substitution. 
\begin{enumerate}
\item If $t \Rightarrow t'$ and $s \Rightarrow s'$, then 
	$\subst t x s \Rightarrow \subst {t'} x {s'}$.
\item If $t \Rightarrow t'$ and $E \Rightarrow E'$, then 
	$\subst t \alpha {\beta E} \Rightarrow \subst {t'} \alpha {\beta E'}$.
\end{enumerate}
\end{lemma}

\begin{proof}
By induction on the derivation of $t \Rightarrow t'$. We treat some cases.
\begin{enumerate}
\item[(t6)] Let $\cctx F {\mu\gamma.c} \Rightarrow \mu\gamma.\subst {c'} \gamma {\gamma F'}$
	with $c \Rightarrow c'$ and $F \Rightarrow F'$. Now we have 
	$\subst c \alpha {\beta E} \Rightarrow \subst {c'} \alpha {\beta E'}$ and 
	$\subst F \alpha {\beta E} \Rightarrow \subst {F'} \alpha {\beta E'}$ by 
	the induction hypothesis. Therefore we have the following.
	\begin{flalign*}
	\subst {(\cctx F {\mu\gamma.c})} \alpha {\beta E}
		\equiv{} & \cctx {(\subst F \alpha {\beta E})} {\mu\gamma.\subst c \alpha {\beta E}} \\
		\Rightarrow{}& \mu\gamma.\subst {\subst {c'} \alpha {\beta E'}} \gamma {\gamma (\subst {F'} \alpha {\beta E'})} \\
		\equiv{} & \mu\gamma.\subst{\subst {c'} \gamma {\gamma F'}} \alpha {\beta E'} \\
		\equiv{} & \subst{(\mu\gamma.\subst {c'} \gamma {\gamma F'})} \alpha {\beta E'}
	\end{flalign*}
	In the before last step we use a substitution lemma. This is possible because
	$\gamma \notin \FCV{E}$ by the Barendregt convention and thus $\gamma \notin \FCV{E'}$ 
	by Lemma~\ref{lemma:muT_pred_fv}.
\item[(c2)] Let $[\alpha]\cctx F {\mu\gamma.c} \Rightarrow \subst {c'} \gamma {\alpha F'}$
	with $c \Rightarrow c'$ and $F \Rightarrow F'$. Now we have 
	$\subst c \alpha {\beta E} \Rightarrow \subst {c'} \alpha {\beta E'}$ and 
	$\subst F \alpha {\beta E} \Rightarrow \subst {F'} \alpha {\beta E'}$ by 
	the induction hypothesis. Therefore we have the following.
	\begin{flalign*}
	\subst {([\alpha]\cctx F {\mu\gamma.c})} \alpha {\beta E}
		\equiv{} & [\beta]\cctx {E(\subst F \alpha {\beta E})} {\mu\gamma.\subst c \alpha {\beta E}} \\
		\Rightarrow{}& \subst {\subst {c'} \alpha {\beta E'}} \gamma {\beta E'(\subst {F'} \alpha {\beta E'})} \\
		\equiv{} & \subst{\subst {c'} \gamma {\alpha F'}} \alpha {\beta E'} \\
		\equiv{} & \subst{(\subst {c'} \gamma {\alpha F'})} \alpha {\beta E'}
	\end{flalign*}
	In the before last step we use a substitution lemma. This is possible because
	$\gamma \notin \FCV{E}$ by the Barendregt convention and thus $\gamma \notin \FCV{E'}$ 
	by Lemma~\ref{lemma:muT_pred_fv}. 
	\qedhere
\end{enumerate}
\end{proof}

A crucial property of a parallel reduction is that a one step
reduction is an instance of a parallel reduction and that a parallel
reduction is an instance of a multi-step reduction.

\begin{lemma}
\label{lemma:muT_mpred_equiv_multi}
Parallel reduction enjoys the intended behavior. That is:
\begin{enumerate}
\item If $t\to t'$, then $t\Rightarrow t'$.
\item If $t\Rightarrow t'$, then $t\tto t'$.
\end{enumerate}
\end{lemma}

\begin{proof}
The first property is proven by induction on the derivation of $t \to t'$ using that 
parallel reduction is reflexive (Lemma~\ref{lemma:muT_pred_refl}). The second by induction on the 
derivation of $t \Rightarrow t'$ using an obvious substitution lemma for $\tto$.
\end{proof}

To define the complete development of a term $t$, we need to decide
which redexes to contract. This job is non-trivial because
$\Rightarrow$ is very strong: In one step it is able to move a subterm
that is located very deeply in the term to the outside. For example,
consider the command $e$:
\begin{equation}	
	e  \equiv  \cctx {E_n} {\mu\alpha_n.[\alpha_n]\ldots \cctx {E_1} {\mu\alpha_1.[\alpha_1]\cctx {E_0} {\mu\alpha_0.c}} \ldots} 
	\label{equation:muT_pred_strong}
\end{equation}
where all the $\mu \alpha_i.[\alpha_i]E_{i-1}$ are
$\mu\eta$-redexes. That is, $\alpha_i \notin \FCV{E_j}$ for all 
\mbox{$0 \le j < i \le n$} and $\alpha_i \notin \FCV c$ 
for all \mbox{$0 \le i \le n$}. Intuitively one would be urged to contract the
\mbox{(t7)-redexes} immediately. That yields:
$$\cctx {E'_n} {\ldots \cctx {E'_1}{\cctx {E'_0} {\mu\alpha_0.c'}}}$$
given complete developments $E'_i$ of $E_i$ and $c'$ of $c$.  However,
this is not the complete development of $e$. We have
$[\alpha_{i+1}]\cctx {E_i} {\mu\alpha_i.d} \Rightarrow d$ for each
$i$ such that $0 \le i < n$, hence the whole command $e$ reduces to
$c'$. As this example indicates, it
is impossible to determine whether a (t7)-redex should be contracted
without looking more deeply into the term. In order to define the complete
development we introduce a special kind of context consisting of a series of 
nested (t7)-redexes, as in (\ref{equation:muT_pred_strong}). Furthermore, 
we define a case distinction on terms.

\begin{definition}
\label{definition:muT_eta_contexts}
A \lambdamuT{} \emph{$\eta$-context} (or simply: an \emph{$\eta$-context}) 
is defined as follows.
\[
	H \inductive \Box \separator \cctx E {\mu\alpha.[\alpha]H}
		\quad\text{provided that }\alpha \notin \FCV H
\]
\end{definition}

The operation of substitution of a term for the hole in an
\mbox{$\eta$-context} is defined in the usual way. However, since
these contexts contain $\mu$-binders it is important that this
operation is \emph{capture avoiding} for $\mu$-variables. Note also
that---in general---an $\eta$-context is not a context in the sense
of Definition~\ref{definition:context}.

\begin{lemma}
\label{lemma:muT_term_classification}
Each term $t$ is of exactly one of the following shapes.
\begin{itemize}
\newcounter{saveenum}
\item[variable]
	\begin{enumerate}
	\item $x$ 
	\setcounter{saveenum}{\value{enumi}}
	\end{enumerate}
\item[value]
	\begin{enumerate}
	\setcounter{enumi}{\value{saveenum}}
	\item $\natenc n$ 
	\item $\lm x.s$ 
	\setcounter{saveenum}{\value{enumi}}
	\end{enumerate}
\item[redex] 
	\begin{enumerate}	
	\setcounter{enumi}{\value{saveenum}}
	\item $(\lm x.s)r$ 
	\item $\nrecT\ r\ s\ \natenc n$
	\item $\cctx H r$	with $H \not\equiv \Box$ and $r \equiv \cctx E {\lambda x.s}$, 
		$r \equiv \cctx E 0$ or $r \equiv \cctx E x$
	\item $\cctx H {\cctx E {\mu\beta.c}}$ with $c \equiv [\gamma]s$ and $\gamma \neq \beta$,
		or $c \equiv [\beta]s$ and $\beta \in \FCV s$
	\setcounter{saveenum}{\value{enumi}}
	\end{enumerate}
\item[other] 
	\begin{enumerate}
	\setcounter{enumi}{\value{saveenum}}
	\item $sr$ with $s \not\equiv \cctx E {\mu\beta.c}$ and $s \not\equiv \lambda x.t$
	\item $\nrecT\ r\ s\ u$ with $u \not\equiv \cctx E {\mu\beta.c}$ and $u \not\equiv \natenc n$
	\item $\sucT u$ with $u \not\equiv \cctx E {\mu\beta.c}$ and $u \not\equiv \natenc n$
	\end{enumerate}
\end{itemize}
\end{lemma}

\begin{proof}
We prove that $t$ is always of one of the given shapes by induction on the
structure of $t$. Furthermore, because these shapes are non-overlapping it
is immediate that $t$ is always of exactly one of the given shapes.
\end{proof}

\begin{definition}
\index{Complete development}
\label{definition:muT_complete_development}
The \emph{complete development} $\MPRED t$ of a term $t$ is defined (using the
case distinction established in Lemma~\ref{lemma:muT_term_classification}) as:
\begin{enumerate}
\item $\MPRED x \defined x$
\item $\MPRED {\natenc n} \defined \natenc n$
\item $\MPRED{(\lm x.s)} \defined \lm x.\MPRED s$
\item $\MPRED{((\lm x.s)r)} \defined \subst {\MPRED s} x {\MPRED r}$
\item $\MPRED{(\nrecT\ r\ s\ 0)} \defined \MPRED r$
\item $\MPRED{(\nrecT\ r\ s\ (\sucT \natenc n))} 
	\defined \MPRED s\ \natenc n\ (\nrecT\ \MPRED r\ \MPRED s\ \natenc n)$
\item $\MPRED{(\cctx H r )} \defined \cctx {\MPRED H} {\MPRED r}$
	
	provided that $H \not\equiv \Box$ and $r \equiv \cctx E {\lambda x.s}$, $r \equiv \cctx E 0$	or $r \equiv \cctx E x$.
	\label{item:muT_complete_development_mu1}
\item $\MPRED{(\cctx H {\cctx E {\mu\beta.c}})} \defined \mu\beta.\subst {\MPRED c} \beta {\beta {\MPRED H \MPRED E}}$
	
	provided that $c \equiv [\gamma]s$ and $\gamma \neq \beta$,
	or $c \equiv [\beta]s$ and $\beta \in \FCV s$.
	\label{item:muT_complete_development_mu2}
\item $\MPRED{(sr)} \defined {\MPRED s} {\MPRED r}$

	provided that $s \not\equiv \cctx E {\mu\beta.c}$ and $s \not\equiv \lambda x.t$
\item $\MPRED{(\nrecT\ r\ s\ u)} \defined \nrecT\ {\MPRED r}\ {\MPRED s}\ {\MPRED u}$

 	provided that $u \not\equiv \cctx E {\mu\beta.c}$ and $u \not\equiv \natenc n$
\item $\MPRED{(\sucT u)} \defined \sucT {\MPRED u}$

	provided that $u \not\equiv \cctx E {\mu\beta.c}$ and $u \not\equiv \natenc n$
\end{enumerate}
with the complete development $\MPRED c$ of a command $c$ defined as:
\begin{enumerate}
\item $\MPRED{([\alpha]\cctx E {\mu\beta.c})} \defined \subst {\MPRED c} \beta {\alpha \MPRED E}$
\item $\MPRED{([\alpha]t)} \defined [\alpha]\MPRED t$

	provided that $t \not\equiv \cctx E {\mu\beta.c}$
\end{enumerate}
the complete development $\MPRED E$ of a context $E$ defined as:
\begin{enumerate}
\item $\MPRED{\Box} \defined \Box$
\item $\MPRED{(Et)} \defined \MPRED E \MPRED t$
\item $\MPRED{(\sucT E)} \defined \sucT \MPRED E$
\item $\MPRED{(\nrecT\ r\ s\ E)} \defined \nrecT\ \MPRED r\ \MPRED s\ \MPRED E$
\end{enumerate}
and the complete development $\MPRED H$ of an $\eta$-context $H$ defined as:
\begin{enumerate}
\item $\MPRED{\Box} \defined \Box$
\item $\MPRED{(\cctx E {\mu\alpha.[\alpha]H})} \defined \MPRED E \MPRED H$
\end{enumerate}
\end{definition}

Towards a proof of confluence, we now want to prove the
following property: if $t\Rightarrow t'$, then $t'\Rightarrow \MPRED t$. 
This is proven by induction on the structure of $t$; the most interesting
cases are when $t\equiv \cctx H r $ (case~\ref{item:muT_complete_development_mu1}
of Definition~\ref{definition:muT_complete_development}) or 
$t\equiv \cctx H {\cctx E {\mu\beta.c}}$ (case~\ref{item:muT_complete_development_mu2}
of Definition~\ref{definition:muT_complete_development}). For these cases we need some 
special lemmas.

\begin{lemma}
\label{lemma:muT_confluence_mu1_help}
Let $r$ be a term such that $r \equiv \cctx E {\lambda x.s}$, $r \equiv \cctx E 0$ 
or $r \equiv \cctx E x$, and $H$ an $\eta$-context. If
$[\alpha]\cctx H r \Rightarrow c$ with $\alpha \notin \FCV{\cctx H r}$, then 
$c\equiv [\alpha]s$ with $\cctx H r \Rightarrow s$ and $\alpha \notin \FCV s$.
\end{lemma}

\begin{proof}
By induction on the structure of $H$.
\end{proof}

\begin{lemma}
\label{lemma:muT_confluence_mu1}
Let $r$ be a term such that $r \equiv \cctx E {\lambda x.s}$, $r \equiv \cctx E 0$ 
or $r \equiv \cctx E x$, and $H$ an $\eta$-context such that $H \not\equiv \Box$. 
If $\cctx H r \Rightarrow t$ and for every strict subexpression $e$ of 
$\cctx H r$ we have $e \Rightarrow e'$ implies $e'\Rightarrow \MPRED e$, then 
$t \Rightarrow \cctx {\MPRED H} {\MPRED r}$.
\end{lemma}

\begin{proof}
We have to consider three cases for the reduction $\cctx H r \Rightarrow t$.
\begin{enumerate}
\item[(t4)] Let $\cctx H r \equiv \cctx {E_s} {\cctx {E_r} {\mu \beta. [\beta]\cctx {H_1} r}}
		\Rightarrow \cctx {E_s'} s$ with $E_s$ a singular context such that 
	$E_s \Rightarrow E_s'$, and $\cctx {E_r} {\mu \beta.[\beta]{\cctx {H_1} r}} \Rightarrow s$.
	By assumption we have \mbox{$E_s' \Rightarrow \MPRED {E_s}$} and 
	\mbox{$s \Rightarrow \MPRED{(\cctx {E_r} {\mu \beta. [\beta]{\cctx {H_1} r}})} 
	  \equiv \cctx {\MPRED {E_r}} {\cctx {\MPRED {H_1}} {\MPRED r}}$}. Therefore,
	by Lemma~\ref{lemma:muT_pred_cctx_subst}, we obtain that
	\mbox{$\cctx {E_s'} s \Rightarrow \cctx {\MPRED{E_s}} {\cctx {\MPRED{E_r}}
	  {\cctx {\MPRED{H_1}} {\MPRED r}}} \equiv \MPRED{(\cctx H r)}$}.
\item[(t6)] Let $\cctx H r \equiv \cctx E {\mu \beta. [\beta]\cctx {H_1} r}
		\Rightarrow \mu \beta. \subst c \beta {\beta E'}$ with $E \Rightarrow E'$ and moreover
	\mbox{$[\beta]{\cctx {H_1} r} \Rightarrow c$}. By Lemma
	\ref{lemma:muT_confluence_mu1_help}, we know that $c\equiv [\beta]s$ with 
	$\cctx {H_1} r \Rightarrow s$ and $\beta\notin\FCV s$. So we are in the
	situation
	\[
		\cctx H r \equiv \cctx E {\mu \beta. [\beta] \cctx {H_1} r} \Rightarrow \mu\beta. [\beta] {\cctx {E'} s}
	\]
	with $E\Rightarrow E'$ and $\cctx {H_1} r \Rightarrow s$.  Now $E' \Rightarrow \MPRED E$ 
	and $s \Rightarrow \MPRED{(\cctx {H_1} r)} \equiv \cctx{\MPRED {H_1}}{\MPRED r}$ by
	assumption. Therefore $\mu\beta. [\beta] {\cctx {E'} s} \Rightarrow \cctx{\MPRED E}{\cctx
	{\MPRED {H_1}} {\MPRED r}} \equiv \MPRED{(\cctx H r)}$ by
	Lemma~\ref{lemma:muT_pred_cctx_subst} and rule (t7).
\item[(t7)] Let $\cctx H r \equiv \mu\beta. [\beta]{\cctx {H_1} r} \Rightarrow s$
	with $\cctx {H_1} r \Rightarrow s$. By assumption we have 
	$s \Rightarrow \MPRED{(\cctx {H_1} r)} \equiv \cctx{\MPRED {H_1}}{\MPRED r}$. 
	Therefore $s \Rightarrow \cctx {\MPRED {H_1}} {\MPRED r} \equiv \MPRED{(\cctx H r)}$.
	\qedhere
\end{enumerate}
\end{proof}

\begin{lemma}
\label{lemma:muT_confluence_mu2}
Let $E$ be a context, $H$ an $\eta$-context, $\gamma$ a $\mu$-variable, and let $d$ be 
a command such that $d \equiv [\beta]s$ with $\beta \neq \gamma$ or 
$d \equiv [\gamma]s$ with $\gamma \in \FCV s$. If 
$\cctx H {\cctx E {\mu\gamma.d}} \Rightarrow t$ and for
every strict subexpression $e$ of $\cctx H {\cctx E {\mu \gamma.d}}$ we have
$e \Rightarrow e'$ implies $e' \Rightarrow \MPRED e$, then 
$t \Rightarrow \mu\alpha.\subst {\MPRED d} \gamma {\alpha \MPRED H \MPRED  E}$. 
\end{lemma}

\begin{proof}
We prove this result by simultaneously proving the following three properties 
by induction on the length of $H$.
\begin{enumerate}
\item If $\cctx H {\cctx E {\mu\gamma.d}} \Rightarrow t$, then $\cctx
  {E_2} t \Rightarrow \mu\alpha.\subst {\MPRED d} \gamma {\alpha {\MPRED {E_2}}
    \MPRED H \MPRED E}$. 
\item If $\cctx H {\cctx E {\mu\gamma.d}} \Rightarrow t$, then $[\alpha] \cctx
  {E_2} t \Rightarrow \subst {\MPRED d} \gamma {\alpha {\MPRED {E_2}}
    \MPRED H \MPRED E}$. 
\item If $[\alpha]\cctx H {\cctx E {\mu\gamma.d}} \Rightarrow c$, then
  $c \Rightarrow \subst {\MPRED d} \gamma {\alpha \MPRED H \MPRED E}$.
\end{enumerate}	

The base case is where $H \equiv \Box$. We only treat a number of
instances for the step case, so let $H \equiv \cctx {E_1} {\mu
  \beta.[\beta] H_1}$.
\begin{enumerate}
\item Let $\cctx {E_1} {\mu \beta.[\beta]\cctx {H_1} {\cctx E  {\mu\gamma.d}}} \Rightarrow t$. 
	Analyzing the possible steps we prove that for every context $E_2$ we have:
	\[
		\cctx {E_2} t 
			\Rightarrow 
		\mu\alpha.\subst {\MPRED d} \gamma {\alpha \MPRED {E_2}\MPRED {E_1} \MPRED {H_1} \MPRED E}.
	\]
	\begin{itemize}
	\item[(t4)] Let $E_1 \equiv E_s E_r$ where $E_s$ is a singular context
		and let $t \equiv \cctx {E_s'}{s}$ with $E_s \Rightarrow E_s'$ and 
		$\cctx {E_r}{\mu \beta.[\beta] \cctx {H_1} {\cctx E {\mu \gamma.d}}} \Rightarrow s$. 
		We can apply the induction  hypothesis for property (1) to $\cctx {E_r}{\mu \beta.[\beta] H_1}$. 
		Now we find that for every context $E_2$ we have:
		\[
			\cctx {E_2}{\cctx {E_s'} s} 
				\Rightarrow 
			\mu\alpha.\subst {\MPRED d} \gamma {\alpha \MPRED {E_2}\MPRED {E_s}\MPRED {E_r} \MPRED {H_1} \MPRED E}.
		\]
	\item[(t6)] Let $\cctx {E_1} {\mu \beta.[\beta] \cctx {H_1}{\cctx{E}{\mu
		  \gamma.d}}} \Rightarrow \mu \beta. \subst c \beta {\beta E_1'}$
		with $E_1 \Rightarrow E_1'$ and $[\beta] \cctx {H_1}{\cctx E{\mu \gamma.d}} \Rightarrow c$. 
		The induction hypothesis for property (3) yields 
		\mbox{$c \Rightarrow \subst{\MPRED d}\gamma {\beta \MPRED {H_1}\MPRED E}$}. 
		Using the substitution Lemma~\ref{lemma:muT_pred_subst} and the rule (t6), 
		we conclude that for any context $E_2$ we have:
		\[
			\cctx {E_2}{\mu \beta. \subst c \beta {\beta E_1'}} 
				\Rightarrow 
			\mu \beta.\subst{\MPRED d}\gamma {\beta \MPRED {E_2} \MPRED {E_1} \MPRED {H_1}\MPRED E}.
		\]
	\item[(t7)] Let $E_1 \equiv \Box$ and 
		$\mu \beta.[\beta] \cctx {H_1}{\cctx E {\mu \gamma.d}} \Rightarrow s$ with 
		$\cctx {H_1}{\mu \gamma.d} \Rightarrow s$. The induction hypothesis for 
		property (1) applied to $\cctx {H_1}{\cctx E {\mu \gamma.d}}$ tells us that for any 
		context $E_2$ we have:
		\[
			\cctx {E_2} s 
				\Rightarrow 
			\mu \gamma . \subst{\MPRED d}\gamma {\gamma \MPRED{E_2}\MPRED{H_1}\MPRED E}.
		\]
	\end{itemize}
\item A similar argument to the one used for (1) also proves (2). 
\item Let $[\alpha]\cctx {E_1}{\mu \beta.[\beta] \cctx {H_1} {\cctx E {\mu \gamma .d}}} \Rightarrow c$. 
	Analyzing the possible steps we prove that we have:
	\[
		c\Rightarrow \subst {\MPRED d} \gamma {\alpha \MPRED {E_1}\MPRED {H_1} \MPRED E}.
	\]
	\begin{itemize}
	\item[(c1)] Let $[\alpha]\cctx {E_1}{\mu \beta.[\beta] \cctx {H_1}
		  {\cctx E {\mu \gamma .d}}} \Rightarrow [\alpha]s$ with 
		\mbox{$\cctx {E_1}{\mu \beta.[\beta] \cctx {H_1} {\cctx E {\mu \gamma .d}}} \Rightarrow s$}.
		To close this case, we have to make a finer case analysis of the
		possible steps that have led to $s$. This is similar to what we have
		done for property (1) above. To close the case we also need the induction
		hypothesis for property (1) and property (2).

	\item[(c2)] Let $[\alpha]\cctx {E_1}{\mu \beta.[\beta] \cctx {H_1}
		  {\cctx E {\mu \gamma .d}}} \Rightarrow \subst c \beta {\alpha E_1'}$
		with $E_1 \Rightarrow E_1'$ and $[\beta] \cctx {H_1} {\cctx E {\mu \gamma .d}} \Rightarrow c$. 
		We apply the induction hypothesis for property (3) to conclude that 
		$c\Rightarrow \subst {\MPRED d} \gamma {\beta \MPRED {H_1} \MPRED E}$. 
		Therefore we have \mbox{$\subst c \beta {\alpha E_1'} \Rightarrow 
		  \subst {\MPRED d} \gamma {\alpha \MPRED{E_1} \MPRED {H_1} \MPRED E}$}
		by the substitution Lemma~\ref{lemma:muT_pred_subst} and we are done. \qedhere
	\end{itemize}
\end{enumerate}
\end{proof}

\begin{theorem}
\label{theorem:muT_maxpred_confluent}
If $t_1 \Rightarrow t_2$, then $t_2 \Rightarrow \MPRED{t_1}$.
\end{theorem}

\begin{proof}
We prove this result by mutual induction on the structure of
terms, commands and contexts.
We use the case distinction made in Lemma~\ref{lemma:muT_term_classification}. 
We consider some interesting cases.
\begin{enumerate}
\item Let $t_1 \equiv x$. In this case just reduction (t1) is possible,
	so $x \Rightarrow \MPRED x \equiv x$.
\item Let $t_1 \equiv (\lm x.s_1)r_1$. In this case the following reductions are possible.
	\begin{enumerate}
	\item[(t4)] $(\lm x.s_1)r_1 \Rightarrow (\lm x.s_2)r_2$ with
		$s_1 \Rightarrow s_2$ and $r_1 \Rightarrow r_2$. Now we have 
		\mbox{$s_2 \Rightarrow \MPRED{s_1}$} and $r_2 \Rightarrow \MPRED{r_1}$ by the 
		induction hypothesis. Therefore we have
		\mbox{$(\lm x.s_2)r_2 \Rightarrow \MPRED{((\lm x.s_1)r_1)} \equiv \subst {\MPRED{s_1}} x {\MPRED{r_1}}$}.
	\item[(t5)] $(\lm x.s_1)r_1 \Rightarrow \subst {s_2} x {r_2}$ with
		$s_1 \Rightarrow s_2$ and $r_1 \Rightarrow r_2$. Now we have 
		$s_2 \Rightarrow \MPRED{s_1}$ and $r_2 \Rightarrow \MPRED{r_1}$ by the 
		induction hypothesis. Therefore
		\mbox{$\subst {s_2} x {r_2} \Rightarrow \MPRED{((\lm x.s_1)r_1)} \equiv \subst {\MPRED{s_1}} x {\MPRED{r_1}}$}
		by Lemma~\ref{lemma:muT_pred_subst}.
	\end{enumerate}
\item Let $t_1 \equiv \cctx {H_1} {r_1}$ with $H_1 \neq \Box$ 
	and $r_1 \equiv \cctx E {\lambda x.s}$, $r_1 \equiv \cctx E 0$ or 
	$r_1 \equiv \cctx E x$. Suppose $t_1 \Rightarrow t_2$.
        Then $t_2 \Rightarrow \cctx {\MPRED{H_1}}
        {\MPRED{r_1}} \equiv \MPRED{t_1}$ by Lemma~\ref{lemma:muT_confluence_mu1}.
\item Let $t_1 \equiv \cctx {H_1} {\cctx {E_1} {\mu\beta.c_1}}$ with
	$c_1 \equiv [\gamma]s$ and $\gamma \neq \beta$, or $c_1 \equiv [\beta]s$ and 
	$\beta \in \FCV s$. 
	Suppose $t_1 \Rightarrow t_2$, then $t_2 \Rightarrow \mu
        \alpha.\subst {\MPRED{c_1}} \beta {\alpha \MPRED{H_1}
          \MPRED{E_1}} \equiv \MPRED{t_1}$ by Lemma
       ~\ref{lemma:muT_confluence_mu2}.
\item Let $t_1 \equiv s_1 r_1$ with $s_1 \not\equiv \cctx E {\mu\alpha.c}$ and
	$s_1 \not\equiv \lambda x.s$. In this case just reduction (t4) is possible, so 
	$s_1 r_1 \Rightarrow s_2 r_2$ with \mbox{$s_1 \Rightarrow s_2$} and $r_1 \Rightarrow r_2$.
	Now $s_1 \Rightarrow \MPRED{s_2}$ and $r_2 \Rightarrow \MPRED{r_1}$ 
	by the induction hypothesis, so $s_2 r_2 \Rightarrow \MPRED{(s_1 r_1)} \equiv \MPRED{s_1}\MPRED{r_1}$. \qedhere
\end{enumerate}
\end{proof}

\begin{corollary}
\index{Confluence!for \lambdamuT{}}
\label{corollary:muT_pred_confluent}
Parallel reduction satisfies the diamond property. That is, if $t_1 \Rightarrow t_2$ and 
\mbox{$t_1 \Rightarrow t_3$}, then there exists a term $t_4$ such that $t_2 \Rightarrow t_4$
and $t_3 \Rightarrow t_4$.
\end{corollary}

\begin{proof}
Let $t_4 = \MPRED{t_1}$. Now we have $t_2 \Rightarrow \MPRED{t_1}$ and
$t_3 \Rightarrow \MPRED{t_1}$ by Theorem~\ref{theorem:muT_maxpred_confluent}.
\end{proof}

\begin{theorem}
Reduction on \lambdamuT{} is confluent. That is, if $t_1 \tto t_2$ and 
$t_1 \tto t_3$, then there exists a term $t_4$ such that $t_2 \tto t_4$
and $t_3 \tto t_4$.
\end{theorem}

\begin{proof}
By Corollary~\ref{corollary:muT_pred_confluent} and the fact that $t
\Rightarrow^* t'$ if and only if $t \tto t'$, which follows
immediately from Lemma~\ref{lemma:muT_mpred_equiv_multi}.
\end{proof}

%% file: arithmetic_sn.tex
\section{Strong normalization of \texorpdfstring{\lambdamuT}{lambda-mu-T}}
\label{section:muT_sn}
\index{Strong normalization!for \lambdamuT}

In this section we prove that the \lambdamucal{T} is strongly normalizing. 
Unfortunately we cannot use the CPS-translation as defined in 
Section~\ref{section:muT_cps} to prove this result. Our CPS-translation 
merely preserves typing and convertibility whereas it does not preserve 
reduction. Defining a CPS-translation that is strictly reduction preserving 
(each reduction step corresponds to one or more reduction steps under the 
translation) is already non-trivial for the \lambdamucal{}, as Ikeda and 
Nakazawa~\cite{ikeda2006} have shown. We failed to extend their approach to 
\lambdamuT{} due to difficulties translating the $\nrecT$ construct.

Instead we prove strong normalization by defining two reductions $\to_A$ and 
$\to_B$ such that $\to\, =\, \to_{AB}\, \defined\, \to_A \cup \to_B$. In 
Section~\ref{section:muT_snA} we prove, using the reducibility method, 
that $\to_A$ is strongly normalizing. In Section~\ref{section:muT_snB}
we prove that $\to_B$ is strongly normalizing and that both reductions 
commute in a way that we can obtain strong normalization for $\to_{AB}$.

To prove strong normalization of the second order call-by-value \lambdamucal{},
Nakazawa~\cite{nakazawa2003} characterizes reductions whose strictness is 
preserved by a modified CPS-translation. Nakazawa also uses a
postponement argument, but the proof is very different from ours.

\begin{definition}
Let $\to_A$ denote the compatible closure of the reduction rules $\to_\beta$, 
$\to_{\mu\sucT}$, $\to_{\mu R}$, $\to_0$, $\to_\sucT$ and $\to_{\mu\natTypeT}$. 
Let $\to_B$ denote the compatible closure of the reduction rules $\to_{\mu\eta}$ 
and $\to_{\mu i}$.
\end{definition}

\begin{definition}
\label{definition:muT_sn}
Given a notion of reduction $\to_X$ (e.g. $\to_A$ or $\to_B$), the set of 
 \emph{strongly normalizing terms}, notation $\SNmuT[X]$, is inductively defined
as follows.
\begin{enumerate}
\item If for all terms $t'$ with $t \to_X t'$ we have $t' \in \SNmuT[X]$, then
	 $t \in \SNmuT[X]$.
\end{enumerate}
\end{definition}

\begin{fact}
\label{fact:muT_nf_sn}
If $t$ is in $\to_X$-normal form, then $t \in \SNmuT[X]$.
\end{fact}

\begin{fact}
\label{fact:muT_steps_sn}
If $t \in \SNmuT[X]$ and $t \tto_X t'$, then $t' \in \SNmuT[X]$.
\end{fact}

\subsection{Strong normalization of \texorpdfstring{\(\to_{A}\)}{(A)}}
\label{section:muT_snA}
In this subsection we prove that $\to_A$-reduction is strongly normalizing 
using the reducibility method. Our proof is inspired by Parigot's 
proof of strong normalization for the \lambdamucal{}~\cite{parigot1997}.

Since we only consider
$\to_A$-reduction we will omit subscripts from all notations. Moreover, 
for conciseness of notation we specify most of the forthcoming 
lemmas only for terms and not for commands. 

\index{Reducibility method}
The reducibility method is originally due to Tait~\cite{tait1967}, who 
proposed the following interpretation for $\to$-types.
\begin{flalign*}
\muTintp{\alpha} &\defined \SNmuT \\
\muTintp{\sigma \to \tau} &\defined \{ t \separator \forall s \in \muTintp\sigma\ .\ ts \in \muTintp\tau \}
\end{flalign*}
This interpretation makes it possible to prove strong normalization of $\lambda\Arrow$ 
in a very short and elegant way~\cite[for example]{geuvers2008}. Instead of 
proving that a term $t$ of type $\rho$ is strongly normalizing one proves a slight generalization, 
namely $t \in \muTintp\rho$. This method also extends to 
\lambdaT{}~\cite[for example]{girard1989}.

Unfortunately, for \lambdamu{} it becomes more complicated. If a term of the shape 
$\lambda x.r$ consumes an argument, the $\lambda$-abstraction vanishes. However, if a term 
of the shape $\mu\alpha.c$ consumes an argument the $\mu$-abstraction remains, hence it is
not possible to predict how many arguments $\mu\alpha.c$ will consume. To repair this issue
Parigot has proposed a way to switch between a term that is a member of a certain 
 \emph{reducibility candidate} and one that is strongly normalizing when applied to 
a certain set of sequences of arguments. 

In \lambdamuT{} a term of the shape $\mu\alpha.c$ is not only able to 
consume arguments on its right hand side, but is also able to consume an unknown
number of $\sucT$'s and $\nrecT$'s. Therefore we generalize Parigot's idea to 
contexts so that we are able to switch between a term that is a member of a certain 
reducibility candidate and one that is strongly normalizing in a certain set of contexts.

Before going into the details of the proof we state some facts.

\begin{fact}
\label{fact:muT_redA_bound}
If $t \in \SNmuT$, then we have that the length of each $\to_A$-reduction 
sequence starting at $t$ is bounded. We use the notation $\MRP t$ to denote this bound.
\end{fact}

\begin{proof}
The result holds because $\to_A$-reduction is finitely branching.
\end{proof}

\begin{fact}
If $t \in \SNmuT$ and $t \to t'$, then $\MRP {t'} < \MRP t$.
\end{fact}

\begin{fact}
\label{fact:muT_redA_subst}
\label{fact:muT_redA_strucsubst}
\(\to_A\)-reduction is preserved under (structural) substitution. 
\begin{enumerate}
\item If $t \to t'$, then $\subst t x s \to \subst {t'} x s$.
\item If $s \to s'$, then $\subst t x s \tto \subst t x {s'}$.
\item If $t \to t'$, then $\cctx E t \to \cctx E {t'}$ and 
	$\subst t \alpha {\beta E} \to \subst {t'} \alpha {\beta E}$.
\item If $E \to E'$, then $\cctx E t \to \cctx {E'} t$ and 
	$\subst t \alpha {\beta E} \tto \subst t \alpha {\beta E'}$.
\end{enumerate}
\end{fact}

We now extend the notion of strongly normalizing terms to strongly
normalizing contexts. Informally a context is strongly normalizing if
all its sub-terms are strongly normalizing.

\index{Strongly normalizing contexts}
\begin{definition}
The set of \emph{strongly normalizing contexts}, notation $\SNmuTcctx$, is
inductively defined as follows.
\begin{enumerate}
\item $\Box \in \SNmuTcctx$
\item If $E \in \SNmuTcctx$ and $t\in \SNmuT$, then $Et \in \SNmuTcctx$.
\item If $E \in \SNmuTcctx$, then $\sucT E \in \SNmuTcctx$.
\item If $E \in \SNmuTcctx$, $r \in \SNmuT$ and $s \in \SNmuT$, then $\nrecT\ r\ s\ E \in \SNmuTcctx$.
\end{enumerate}
\end{definition}

Parigot's approach has another advantage; for the expansion lemmas we do not 
need to worry about the interpretation of types. We merely 
need the notion of being strongly normalizing (with respect to some context).

\begin{lemma}
\label{lemma:muT_snA_context_redex}
\label{lemma:muT_snA_context_zero}
Let $E$ be a context and $r$ a term such that $r \equiv x$, $r \equiv (\lm x.r)t$, 
$r \equiv \nrecT\ r\ s\ \natenc n$ or $r \equiv \cctx {E^s} {\mu \alpha.c}$.
If $\cctx E r \to t$, then we have:
\begin{enumerate}
\item $t \equiv \cctx E {r'}$ with $r \to r'$, or,
\item $t \equiv \cctx {E'} r$ with $E \to E'$.
\end{enumerate}
\end{lemma}

\begin{proof}
We prove the result by induction on the structure of $E$. We consider only the
case $E \equiv Ft$. Here we use the assumption about the shape of $r$ to derive 
that $\cctx F r$ cannot be of the shape $\lambda x.s$ or $\mu\beta.c$. This 
guarantees that $\cctx F r t$ is not a redex, by which the result follows
immediately.
\end{proof}

\begin{lemma}
\label{lemma:muT_snA_lambda}
If $r \in \SNmuT$ and \mbox{$\cctx E {\subst t x r} \in \SNmuT$}, then $\cctx E {(\lm x.t)r} \in \SNmuT$.
\end{lemma}

\begin{proof}
We use Fact~\ref{fact:muT_redA_bound} to prove this result by well-founded 
induction on $\MRP r + \MRP{\cctx E {\subst t x r}}$. By 
Definition~\ref{definition:muT_sn} we have to show that for each term $w$
with $\cctx E {(\lm x.t) r} \to w$ we have $w \in \SNmuT$.
\begin{enumerate}
\item Let $w \equiv \cctx E {\subst t x r}$. Now 
	$\cctx E {\subst t x r} \in \SNmuT$ by assumption.
\item Let $w \equiv \cctx E {(\lm x.t')r}$ and $t \to t'$. Now 
	$\cctx E {\subst t x r} \to \cctx E {\subst {t'} x r}$ by 
	Fact~\ref{fact:muT_redA_subst}, hence $\cctx E {\subst {t'} x r} \in \SNmuT$. 
	By the induction hypothesis we have $\cctx E {(\lm x.t')r} \in \SNmuT$
	since $\MRP{\cctx E {\subst {t'} x r}} < \MRP{\cctx E {\subst t x r}}$.
\item Let $w \equiv \cctx E {(\lm x.t)r'}$ and $r \to r'$. Now 
	$\cctx E {\subst t x r} \tto \cctx E {\subst t x {r'}}$ by Fact~\ref{fact:muT_redA_subst}
	and therefore $\cctx E {\subst t x {r'}} \in \SNmuT$. By the induction 
	hypothesis we have $\cctx E {(\lm x.t)r'} \in \SNmuT$ since $\MRP {r'} < \MRP r$.
\item Let $w \equiv \cctx E {(\lm x.t)r}$ and $E \to E'$. Now
	$\cctx E {\subst t x r} \to \cctx {E'} {\subst t x r}$ by 
	Fact~\ref{fact:muT_redA_subst}, hence $\cctx {E'} {\subst t x r} \in \SNmuT$. 
	By the induction hypothesis we have $\cctx {E'} {(\lm x.t)r} \in \SNmuT$
	since $\MRP{\cctx {E'} {\subst t x r}} < \MRP{\cctx E {\subst t x r}}$.
\end{enumerate}
Lemma~\ref{lemma:muT_snA_context_redex} guarantees that we have considered all possible shapes of $w$.
\end{proof}

\begin{lemma}
\label{lemma:muT_snA_mu}
If $F^s \in \SNmuTcctx$ and $\cctx E {\mu\alpha.\subst c \alpha {\alpha F^s}} \in \SNmuT$,
then \mbox{$\cctx E {\cctx {F^s} {\mu \alpha.c}} \in \SNmuT$}.
\end{lemma}

\begin{proof}
The proof is similar to the proof of Lemma~\ref{lemma:muT_snA_lambda}.
\end{proof}

\begin{corollary}
\label{corollary:muT_snA_mu}
If $F \in \SNmuTcctx$ and $\cctx E {\mu\alpha.\subst c \alpha {\alpha F}} \in \SNmuT$,
then $\cctx E {\cctx F {\mu\alpha.c}} \in \SNmuT$.
\end{corollary}

\begin{proof}
By induction on the structure of $F$.
\begin{enumerate}
\item Let $F \equiv \Box$. We have
	$\cctx E {\mu\alpha.c} \equiv \cctx E {\mu\alpha.\subst c \alpha {\alpha \Box}}$
	for each context $E$ and command $c$, so by assumption we are done.
\item Let $F \equiv G^sH$. By an obvious substitution lemma and 
	assumption we have \mbox{$\cctx E {\mu\alpha.\subst {\subst c \alpha {\alpha H}} \alpha {\alpha G^s}} 
		\equiv \cctx E {\mu\alpha.\subst c \alpha {\alpha F}} \in \SNmuT$}.
	Therefore we have $\cctx E {\cctx {G^s} {\mu\alpha.\subst c \alpha {\alpha H}}} \in \SNmuT$ by 
	Lemma~\ref{lemma:muT_snA_mu}. Hence
	$\cctx E {\cctx {G^s} {\cctx H {\mu\alpha.c}}} \in \SNmuT$
	by the induction hypothesis.
\qedhere
\end{enumerate}
\end{proof}

\begin{lemma}
\label{lemma:muT_snA_nrec_exp}
For each context $E$ we have the following.
\begin{enumerate}
\item If $\cctx E r \in \SNmuT$ and $s \in \SNmuT$, then $\cctx E {\nrecT\ r\ s\ 0} \in \SNmuT$.
\item If $\cctx E {s\ \natenc n\ (\nrecT\ r\ s\ \natenc n)} \in \SNmuT$, then $\cctx E {\nrecT\ r\ s\ (\sucT \natenc n)} \in \SNmuT$.
\end{enumerate}
\end{lemma}

\begin{proof}
We use Fact~\ref{fact:muT_redA_bound} and prove (1) by 
induction on $\MRP {{\cctx E r}} + \MRP s$ and (2) by 
induction on $\MRP {\cctx E {s\ \natenc n\ (\nrecT\ r\ s\ \natenc n)}}$.
Similar to the proof of Lemma~\ref{lemma:muT_snA_lambda} we distinguish various
cases.
\end{proof}

Parigot extends the well-known \emph{functional construction} of two
sets of terms $S$ and $T$ ($S \to T \defined \{ t \separator \forall u
\in S \ .\ tu \in T \}$) to a set $\mathcal S$ of sequences of terms
and a set $T$ of terms as follows.
\[
	\mathcal S \to T \defined \{ t \separator \forall \vec u \in \mathcal S \ .\ t\vec u \in T \}
\]
Moreover, he defines the notion of reducibility candidates in such way
that each reducibility candidate $R$ can be expressed as $\mathcal S
\to \SNmuT$ for a certain set of sequences of terms $\mathcal S$.
Therefore he is able to switch between the proposition $t \in R$ and
the proposition $t\vec u \in \SNmuT$ for all $\vec u \in \mathcal
S$. We extend Parigot's notion of functional construction to contexts
in the obvious way.

\index{Functional construction}
\begin{definition}
\label{def:muT_snA_func_cons}
Given a set of contexts $\mathcal E$ and a set of terms $T$, the \emph{functional 
construction} $\mathcal E \to T$ is defined as follows.
\[
	\mathcal E \to T \defined \{ t \separator \forall E \in \mathcal E\ .\ \cctx E t \in T \}
\]
Given two sets of terms  $S$ and $T$, then $S \to T$ is defined as follows.
\[
	S \to T \defined \{ \Box u \separator u \in S \} \to T
\]
\end{definition}

Remark that, for sets of terms $S$ and $T$, our definition of the
functional construction $S \to T$ is equivalent to the ordinary
definition.
\[
	S \to T = \{ \Box u \separator u \in S \} \to T = \{ t \separator \forall u \in S \ .\ tu \in T \}
\]

Keeping in mind that we wish to express each reducibility candidate $R$ as 
$\mathcal E \to \SNmuT$ for some $\mathcal E$, one might try to  
define the collection of reducibility candidates as the smallest set that contains 
$\SNmuT$ and is closed under functional construction and arbitrary intersection.
But then $\{ \nrecT\ \Omega\ \Omega\ \Box \} \to \SNmuT = \emptyset$ is a valid
candidate too. To avoid this we should be a bit more careful.

\begin{definition}
\index{Reducibility candidates}
\label{def:muT_snA_red_cand}
We define the collection of \emph{reducibility candidates},
$\SNmuTredcand$, inductively as follows.
\begin{enumerate}
\item[\em (sn)] $\SNmuT \in \SNmuTredcand$
\item[\em ($\bigcap$)] If $\emptyset \subset \mathbf R \subseteq \SNmuTredcand$, then $\bigcap \mathbf R \in \SNmuTredcand$.
\item[\em (app)] If $S, T \in \SNmuTredcand$, then $S \to T \in \SNmuTredcand$.
\item[\em (suc)] If $T \in \SNmuTredcand$, then $\{\sucT \Box \} \to T \in \SNmuTredcand$.
\item[\em (nrec)] If $S,T \in \SNmuTredcand$, then $\{\nrecT\ r\ s\ \Box \separator r \in T, s \in S \to T \to T \} \to T \in \SNmuTredcand$.
\end{enumerate}
\end{definition}

\begin{lemma}
\label{lemma:muT_snA_var}
For each $R \in \SNmuTredcand$ we have the following.
\begin{enumerate}
\item $R \subseteq \SNmuT$
\item $\cctx E x \in R$ for each $x$ and $E \in \SNmuTcctx$.
\end{enumerate}
\end{lemma}

\begin{proof}
We prove these results simultaneously by induction on the generation of $R$. 
We consider some interesting cases.
\begin{enumerate}
\item[(sn)] Let $R = \SNmuT$. We certainly have $R \subseteq \SNmuT$.
	Also, $\cctx E x \in \SNmuT$ by Lemma~\ref{lemma:muT_snA_context_zero}.
\item[($\bigcap$)] Let $R = \bigcap \mathbf R$. By the induction hypothesis we 
	have $T \subseteq \SNmuT$ for each $T \in \mathbf R$. Therefore we have
	$\bigcap \mathbf R  \subseteq \SNmuT$, so the first property holds.
	
	By the induction hypothesis we also have $\cctx E x \in T$ for each 
	$T \in \mathbf R$ and $E \in \SNmuTcctx$. Therefore we have $\cctx E x \in \mathbf R$
	for each $E \in \SNmuTcctx$, so the second property holds as well.
\item[(suc)] Let $R = \{\sucT \Box \} \to T$. To prove the first property, 
	we suppose that $t \in R$. This means that $\sucT t \in T$. Therefore
	$\sucT t \in \SNmuT$ because $T \subseteq \SNmuT$ by the induction
	hypothesis. Now certainly $t \in \SNmuT$, so the first property holds.
	
	To prove the second property we have to show that $\cctx E x \in R$. 
	By the induction hypothesis we have $\cctx E x \in T$ 
	for each $E \in \SNmuTcctx$. In particular we have $\cctx {\sucT E} x \in T$.
	This means that $\cctx E x \in R$, so the second property holds as well.
\item[(nrec)] Let $R = \{\nrecT\ r\ s\ \Box \separator r \in T, s \in S\to T\to T \} \to T$. 
	To prove the first property, we suppose that $t \in R$. This means that 
	\mbox{$\nrecT\ r\ s\ t \in T$} for each $r \in T$ and $s \in S\to T\to T$. 
	By the induction hypothesis we have an $x \in T$ and $y \in S\to T\to T$, 
	hence $\nrecT\ x\ y\ t \in T$. Thus $t \in \SNmuT$ because $T \subseteq \SNmuT$ 
	by the induction hypothesis, so the first property holds. 
	
	To prove the second property we have to show that $\cctx E x \in R$. 
	By the induction hypothesis we have $\cctx E x \in T$ 
	for each $E \in \SNmuTcctx$. In particular we have $\cctx {\nrecT\ r\ s\ E} x \in T$.
	This means that $\cctx E x \in R$, so the second property holds as well. \qedhere
\end{enumerate}
\end{proof}

As we have remarked before, we wish to express each reducibility candidate 
$R$ as $\mathcal E \to \SNmuT$ for some set of contexts $\mathcal E$. Now we will make
that idea precise.

\begin{definition}
Given an $R \in \SNmuTredcand$, a set of contexts $\SNmuTreddual R$ is 
inductively defined on the generation of $R$ as follows.
\begin{flalign*}
\SNmuTreddual\SNmuT & \defined{}
	\{ \Box \} \\
\SNmuTreddual{(\bigcap \mathbf R)} & \defined{} 
	\bigcup\{ \SNmuTreddual T \separator T \in \mathbf R \} \\
\SNmuTreddual{(S \to T)} &\defined{} 
	\{ \Box \} \cup \{ E(\Box u) \separator u \in S, E \in \SNmuTreddual T \}  \\
\SNmuTreddual{(\{\sucT \Box\} \to T)} &\defined{}
	\{ \Box \} \cup \{ E(\sucT \Box) \separator E \in \SNmuTreddual T \} \\
\SNmuTreddual{(\{\nrecT\ r\ s\ \Box\} \to T)} &\defined{}
	\{ \Box \} \cup \{ E(\nrecT\ r\ s\ \Box) \separator r \in T, s \in S\to T\to T,E \in \SNmuTreddual T \}
\end{flalign*}
\end{definition}

\begin{fact}
\label{fact:muT_snA_box}
For each $R \in \SNmuTredcand$ we have $\Box \in \SNmuTreddual R$.
\end{fact}

\begin{lemma}
\label{lemma:muT_snA_to_sn}
For each $R \in \SNmuTredcand$ we have $R = \SNmuTreddual R \to \SNmuT$.
\end{lemma}

\begin{proof}
By induction on the generation of $R$. We consider some interesting
cases.
\begin{enumerate}
\item[(sn)] Let $R = \SNmuT$. We have $R = \{\Box\} \to \SNmuT$, so we are done.
\item[($\bigcap$)] Let $R = \bigcap \mathbf R$. By the induction hypothesis we 
	have $T = \SNmuTreddual T \to \SNmuT$ for each $T \in \mathbf R$.
	Therefore we have the following.
	\begin{flalign*}
	R
		&= \bigcap \{ T \separator T \in \mathbf R \} \\
	 	&= \bigcap \{ \SNmuTreddual  T \to \SNmuT \separator T \in \mathbf R \} \\
	 	&= \bigcap \{ \{ t \separator \forall E \in \SNmuTreddual T\ .\ \cctx E t \in \SNmuT \}
	 		\separator T \in \mathbf R \} \\
	 	&= \{ t \separator \forall T \in \mathbf R, E \in \SNmuTreddual T\ .\ \cctx E t \in \SNmuT \} \\
	 	&= \{ t \separator \forall E \in 
	 		\bigcup \{ \SNmuTreddual T \separator T \in \mathbf R \}\  .\ \cctx E t \in \SNmuT \} \\
	 	&= \bigcup\{ \SNmuTreddual T \separator T \in \mathbf R \} \to \SNmuT
	 \end{flalign*}
\item[(nrec)] Let $R = \{\nrecT\ r\ s\ \Box \separator r \in T, s \in S\to T\to T \} \to T$. 
	By the induction hypothesis we have $T = \SNmuTreddual T \to \SNmuT$. Therefore we 
	have the following.
	\begin{flalign*}
	 R
	 	&= \{\nrecT\ r\ s\ \Box \separator r \in T, s \in S\to T\to T \} \to T \\
	 	&= \{\nrecT\ r\ s\ \Box \separator r \in T, s \in S\to T\to T\} \to \SNmuTreddual T \to \SNmuT \\
	 	&= \{ t \separator \forall r \in T, s \in S\to T\to T\ .\ \nrecT\ r\ s\ t \in \SNmuTreddual T \to \SNmuT \} \\
	 	&= \{ t \separator \forall E \in \SNmuTreddual T,r \in T, s \in S\to T\to T\ .\ \cctx E {\nrecT\ r\ s\ t} \in \SNmuT \} \\
	 	&= \{ t \separator t \in \SNmuT \land \forall E \in \SNmuTreddual T,r \in T, s \in S\to T\to T\ .\ \cctx E {\nrecT\ r\ s\ t} \in \SNmuT \} \\
	 	&= \big(\{ \Box \} \cup \{ E(\nrecT\ r\ s\ \Box) \separator r \in T, s \in S\to T\to T,E \in \SNmuTreddual T \}\big) \to \SNmuT
	 \end{flalign*}	 
	The before last step holds because for all terms $t$, if
	$\cctx E {\nrecT\ r\ s\ t} \in \SNmuT$ for all $E \in \SNmuTreddual T$, $r \in T$, 
	$s \in S\to T\to T$, then also $t \in \SNmuT$. This is because
	$\SNmuTreddual T$, $T$ and $S\to T\to T$ are non-empty by 
	Fact~\ref{fact:muT_snA_box} and Lemma~\ref{lemma:muT_snA_var}.
	\qedhere
\end{enumerate}
\end{proof}

\begin{lemma}
\label{lemma:muT_snA_switch}
For each $R \in \SNmuTredcand$ we have $t \in R$ iff $\cctx E t \in \SNmuT$
for all $E \in \SNmuTreddual{R}$.
\end{lemma}

\begin{proof}
We have $t \in R$ iff $t \in \SNmuTreddual R \to \SNmuT$ 
by Lemma~\ref{lemma:muT_snA_to_sn}, and $t \in \SNmuTreddual R \to \SNmuT$
iff $\cctx E t \in \SNmuT$ for all $E \in \SNmuTreddual R$ by 
Definition~\ref{def:muT_snA_func_cons}.
\end{proof}

Now, to prove strong normalization of $\to_A$, it remains to give 
an interpretation $\muTintp \rho \in \SNmuTredcand$ for each type $\rho$. 
As a first attempt, we could adapt the definition for $\lambda\Arrow$, which
we have given in the introduction of this section.
\begin{flalign*}
\muTintp{\natTypeT} &\defined \SNmuT \\
\muTintp{\sigma \to \tau} &\defined \muTintp \sigma \to \muTintp \tau
\end{flalign*}
Unfortunately, the interpretation of $\natTypeT$ does not contain enough
structure to prove the following properties.
\begin{enumerate}
\item If $t \in \SNmuT$, then $\sucT t \in \SNmuT$.
\item If $t \in \SNmuT$, $r \in S$ and $s \in \SNmuT \to S \to S$, then
	$\nrecT\ r\ s\ t \in S$.
\end{enumerate}
Here, the term $t$ could reduce to a term of the shape $\mu\alpha.c$ and is
thereby able to consume the surrounding $\sucT$ or $\nrecT$. To
define an interpretation of $\natTypeT$ that contains more structure
we introduce the following definition.

\begin{definition}
\label{lemma:muT_snA_nat_cand}
We define the collection $\mathcal N$ inductively as follows.
\begin{enumerate}
\item[\em (sn)] $\SNmuT \in \mathcal N$
\item[\em (suc)] If $S \in \mathcal N$, then $\{ \sucT \Box\} \to S \in \mathcal N$.
\item[\em (nrec)] If $S \in \mathcal N$ and $T \in \SNmuTredcand$, then 
	\mbox{$\{\nrecT\ r\ s\ \Box \separator r \in T, s \in S\to T\to T \} \to T \in \mathcal N$}.
\end{enumerate}
\end{definition}

\begin{fact}
$\mathcal N \subseteq \SNmuTredcand$
\end{fact}

\begin{definition}
\label{lemma:muT_snA_type_intp}
The \emph{interpretation} $\muTintp \rho$ of a type $\rho$ is defined
as follows.
\begin{flalign*}
\muTintp \natTypeT & \defined \bigcap \mathcal N \\
\muTintp{\sigma \to \tau} &\defined \muTintp \sigma \to \muTintp \tau
\end{flalign*}
\end{definition}

\begin{fact}
\label{fact:muT_snA_intp_cand}
For each type $\rho$ we have $\muTintp \rho \in \SNmuTredcand$.
\end{fact}

\begin{lemma}
\label{lemma:muT_snA_zero}
For each $n \in \nat$ we have $\natenc n \in \muTintp \natTypeT$.
\end{lemma}

\begin{proof}
In order to prove this result we have to show that $\natenc n \in R$ for all 
$R \in \mathcal N$ and $n \in \nat$. We proceed by induction on the generation of $R$. 
\begin{enumerate}
\item[(var)] Let $R = \SNmuT$. Now we have to show that $\natenc n \in \SNmuT$ 
	for all $n \in \nat$. However, $\natenc n$ is in normal form, so we
	certainly have $\natenc n \in \SNmuT$.
\item[(suc)] Let $R = \{ \sucT \Box\} \to S$. Now we have $\natenc n \in S$ for all $n \in \nat$
	by the induction hypothesis. It remains to show that $\sucT{\natenc n} \in S$ 
	for all $n \in \nat$. However, $\sucT{\natenc n} \equiv \natenc{n + 1}$, so the
	required result follows from the induction hypothesis.
\item[(nrec)] Let $R = \{\nrecT\ r\ s\ \Box \separator r \in T, s \in S\to T\to T \} \to T$.
	Now we have $\natenc n \in S$ for all $n \in \nat$ by the induction hypothesis.
	It remains to show that $\nrecT\ r\ s\ \natenc n \in T$ for all 
	$S \in \mathcal N$, $T \in \SNmuTredcand$, $r \in T$, $s \in S \to T \to T$ and $n \in \nat$. 
	We proceed by induction on $n$.
	\begin{enumerate}
	\item Let $n = 0$. We have $\cctx E r \in \SNmuT$ for all $E \in \SNmuTreddual T$
		by Lemma~\ref{lemma:muT_snA_switch} and $s \in \SNmuT$ by Lemma~\ref{lemma:muT_snA_var}.
		Hence $\cctx E {\nrecT\ r\ s\ 0} \in \SNmuT$ by Lemma~\ref{lemma:muT_snA_nrec_exp}
		and therefore $\nrecT\ r\ s\ 0 \in T$ by Lemma~\ref{lemma:muT_snA_switch}.
	\item Let $n > 0$. We have $\nrecT\ r\ s\ \natenc {n-1} \in T$ by the 
		induction hypothesis. Furthermore, because $s \in S \to T \to T$ and $\natenc {n-1} \in S$,
		we have $s\ \natenc {n-1}\ (\nrecT\ r\ s\ \natenc{n-1}) \in T$, so
		$\cctx E {s\ \natenc {n-1}\ (\nrecT\ r\ s\ \natenc {n-1})} \in \SNmuT$ 
		for all $E \in \SNmuTreddual T$ by Lemma~\ref{lemma:muT_snA_switch}. Therefore 
		$\cctx E {\nrecT\ r\ s\ (\sucT \natenc{n-1})} \in \SNmuT$ by 
		Lemma~\ref{lemma:muT_snA_nrec_exp}, so $\nrecT\ r\ s\ \natenc{n} \in T$
		by Lemma~\ref{lemma:muT_snA_switch}. \qedhere
\end{enumerate}
\end{enumerate}
\end{proof}

\begin{lemma}
\label{lemma:muT_snA_suc}
If $t \in \muTintp\natTypeT$, then $\sucT t \in \muTintp\natTypeT$.
\end{lemma}
	
\begin{proof}
Assume that $t \in \muTintp\natTypeT$. This means, $t \in R$ for all 
$R \in \mathcal N$. Now we have to prove that $\sucT t \in R$ for all $R \in \mathcal N$.
But for all $R \in \mathcal N$ we have $\{ \sucT \Box \} \to R \in \mathcal N$,
hence $t \in \{ \sucT \Box \} \to R$ by assumption and therefore $\sucT t \in R$.
\end{proof}

\begin{lemma}
\label{lemma:muT_snA_nrec}
If $r \in \muTintp \rho$, $s \in \muTintp{\natTypeT \to \rho \to \rho}$ and $t \in \muTintp\natTypeT$,
then $\nrecT\ r\ s\ t \in \muTintp \rho$.
\end{lemma}

\begin{proof}
We have $\muTintp\natTypeT \in \mathcal N$ by Definition~\ref{lemma:muT_snA_type_intp}, 
so if $t \in \muTintp\natTypeT$, then $\nrecT\ r\ s\ t \in T$ for all 
$T \in \SNmuTredcand$, $r \in T$ and $s \in \muTintp\natTypeT \to T \to T$
by Definition~\ref{lemma:muT_snA_nat_cand}. Also $\muTintp \rho \in \SNmuTredcand$ 
by Fact~\ref{fact:muT_snA_intp_cand} and  
$\muTintp{\natTypeT\to\rho\to\rho} = \muTintp\natTypeT\to\muTintp\rho\to\muTintp\rho$ ,
hence $\nrecT\ r\ s\ t \in \muTintp \rho$.
\end{proof}

\begin{theorem}
\label{theorem:muT_snA_aux}
Let $\mujudg {x_1 : \rho_1, \ldots, x_n : \rho_n} {\alpha_1 : \sigma_1, \ldots, \alpha_m : \sigma_m} t \tau$ such that
$r_i \in \muTintp{\rho_i}$ for all $1 \le i \le n$ and $E_j \in \SNmuTreddual{\muTintp{\sigma_j}}$ 
for all $1 \le j \le m$, then:
\[
	t[x_1 := r_1, \ldots, x_n := r_n, \alpha_1 := \alpha_1\ E_1, \ldots, \alpha_m := \alpha_m\ E_m] \in \muTintp \tau.
\]
\end{theorem}

\begin{proof}
Abbreviate $\Gamma = x_1 : \rho_1, \ldots, x_n : \rho_n$, $\Delta = \alpha_1 : \sigma_1, \ldots, \alpha_m : \sigma_m$
, with  
\mbox{$t' \equiv t[x_1 := r_1, \ldots, x_n := r_n, \alpha_1 := \alpha_1\ E_1, \ldots, \alpha_m := \alpha_m\ E_m]$}
, and \\
\mbox{$c' \equiv c[x_1 := r_1, \ldots, x_n := r_n, \alpha_1 := \alpha_1\ E_1, \ldots, \alpha_m := \alpha_m\ E_m]$}
. 
Now by mutual induction we prove that $\mujudg \Gamma \Delta t \tau$ implies $t' \in \muTintp \tau$
and that $\musjudg \Gamma \Delta c$ implies $c' \in \SNmuTc$.
\begin{enumerate}
\item[(var)] Let $\mujudg \Gamma \Delta x \sigma$ with $x : \sigma \in \Gamma$.
	Now we have $x' \in \muTintp \sigma$ by assumption.
\item[($\lambda$)] Let $\mujudg \Gamma \Delta {\lm x : \sigma.t} {\sigma \to \tau}$
	with $\mujudg {\Gamma, x : \sigma} \Delta t \tau$. Moreover let 
	$u \in \muTintp\rho$ and $E \in \SNmuTreddual{\muTintp\tau}$. Now we have 
	$\subst {t'} x u \in \muTintp{\tau}$ by the induction hypothesis and so
	$\cctx E {\subst {t'} x u} \in \SNmuT$	by Lemma~\ref{lemma:muT_snA_switch}. 
	Therefore $\cctx E {(\lm x.t')u} \in \SNmuT$
	by Lemma~\ref{lemma:muT_snA_lambda} and hence $(\lm x.t')u \in \muTintp\tau$  
	by Lemma~\ref{lemma:muT_snA_switch}, so $\lm x.t' \in \muTintp{\sigma \to \tau}$ by Definition~\ref{def:muT_snA_func_cons}.
\item[(app)] Let $\mujudg \Gamma \Delta {ts} \tau$ with
	$\mujudg \Gamma \Delta t {\sigma \to \tau}$	and  $\mujudg \Gamma \Delta s \sigma$. 
	Now we have $t' \in \muTintp{\sigma \to \tau} = \muTintp\sigma \to \muTintp\tau$
	and $s' \in \muTintp\sigma$ by the induction hypothesis, hence 
	$t's' \in \muTintp\tau$ by Definition~\ref{def:muT_snA_func_cons}.
\item[(zero)] Let $\mujudg \Gamma \Delta 0 \natTypeT$. Now we have $0 \in \muTintp \natTypeT$ by Lemma~\ref{lemma:muT_snA_zero}.
\item[(suc)] Let $\mujudg \Gamma \Delta {\sucT t} \natTypeT$ with
	$\mujudg \Gamma \Delta t \natTypeT$. Now we have $t' \in \muTintp\natTypeT$ 
	by the induction hypothesis and therefore $\sucT t' \in \muTintp\natTypeT$
	by Lemma~\ref{lemma:muT_snA_suc}.
\item[(nrec)] Let $\mujudg \Gamma \Delta {\nrecT\ r\ s\ t} \rho$ with
	$\mujudg \Gamma \Delta r \rho$, $\mujudg \Gamma \Delta s {\natTypeT\to\rho\to\rho}$ and
	$\mujudg \Gamma \Delta t \natTypeT$. Now we have $r' \in \muTintp \rho$, 
	\mbox{$s' \in \muTintp{\natTypeT \to \rho \to \rho}$}
	and $t' \in \muTintp \natTypeT$ by the induction hypothesis. Therefore
	$\nrecT\ r'\ s'\ t' \in \muTintp\rho$ by Lemma~\ref{lemma:muT_snA_nrec}.
\item[(act)] Let $\mujudg \Gamma \Delta {\mu\alpha : \rho.c} \rho$ with
	$\musjudg \Gamma {\Delta. \alpha : \rho} c$. Moreover let
	\mbox{$E \in \SNmuTreddual{\muTintp\rho}$}. Now we have $\subst {c'} \alpha {\alpha E} \in \SNmuTc$
	by the induction hypothesis. Hence $\mu\alpha.\subst {c'} \alpha {\alpha E} \in \SNmuT$ 
	and therefore $\cctx E {\mu\alpha.c'} \in \SNmuT$ by Corollary~\ref{corollary:muT_snA_mu},
	so $\mu\alpha.c' \in \muTintp \rho$ by Lemma~\ref{lemma:muT_snA_switch}.
\item[(pas)] Let $\musjudg \Gamma \Delta {[\alpha]t}$ with $\alpha : \sigma \in \Delta$
	and $\mujudg \Gamma \Delta t \sigma$. Now we have $t' \in \muTintp \sigma$
	by the induction hypothesis. Also, we have a context 
	\mbox{$E \in \SNmuTreddual{\muTintp\sigma}$}	by assumption. Therefore
	$\cctx E {t'} \in \SNmuT$ by Lemma~\ref{lemma:muT_snA_switch} and so
	$[\alpha]\cctx E {t'} \in \SNmuTc$ because \mbox{$([\alpha]t)' = [\alpha]\cctx E {t'}$}.\qedhere
\end{enumerate}
\end{proof}

\begin{corollary}
\label{corollary:muT_snA}
If $\mujudg \Gamma \Delta t \rho$, then $t \in \SNmuT[A]$.
\end{corollary}

\begin{proof}
We have $x_i \in \muTintp{\rho_i}$ for each $x_i : \rho_i \in \Gamma$ by Lemma~\ref{lemma:muT_snA_var}
and $\Box \in \SNmuTreddual{\muTintp{\sigma_j}}$ for each $\alpha_j : \sigma_j \in \Delta$ by 
Fact~\ref{fact:muT_snA_box}. Therefore $t \in \muTintp \rho$
by Theorem~\ref{theorem:muT_snA_aux} and hence $t \in \SNmuT[A]$ by 
Fact~\ref{fact:muT_snA_intp_cand} and Lemma~\ref{lemma:muT_snA_var}.
\end{proof}

\subsection{Strong normalization of \texorpdfstring{\(\to_{AB}\)}{(AB)}}
\label{section:muT_snB}
In this section we prove that $\to_B$ is strongly normalizing and that 
$\to_A$-steps can be advanced. Together with strong normalization of $\to_A$
this is sufficient to prove strong normalization of $\to_{AB}$. Proving
strong normalization of $\to_{AB}$ from $\to_A$ and $\to_B$ is not specific 
to \lambdamuT{}. Krebbers\footnote{The \Coq{} proof is available at 
	\href{http://robbertkrebbers.nl/misc/sn_commute.html}
		{\texttt{http://robbertkrebbers.nl/misc/sn\_commute.$\{$v,html$\}$}}.} 
provides a proof of this result based on abstract relations in the \Coq{} proof 
assistant.

\begin{lemma}
\label{lemma:muT_snB}
For each term $t$ we have $t \in \SNmuT[B]$.
\end{lemma}

\begin{proof}
By performing a $\to_{\mu\eta}$ or $\to_{\mu i}$-reduction step on $t$, the term
$t$ reduces strictly in its size and therefore $\to_B$-reduction is strongly
normalizing.
\end{proof}

\begin{lemma}
\label{lemma:muT_sn_postponement}
A single $\to_A$-reduction step can be advanced. That means, if \mbox{$t_1 \to_B t_2 \to_A t_3$},
then there is a $t_4$ such that the following diagram commutes.
\[\xymatrix{
	t_1 \ar_A@{.>}[d] \ar^B@{->}[r] & t_2 \ar^A@{->}[d]\\
	t_4 \ar_{AB}@{.>>}[r] & t_3
}\]
\end{lemma}

\begin{proof}
We prove this lemma by distinguishing cases on $t_1 \to_B t_2$ and 
$t_2 \to_A t_3$, we treat some interesting cases.
\begin{enumerate}
\item Let $(\lm x.t)r \to_B (\lm x.t)r' \to_A \lm x.\subst t x {r'}$ with
	$r \to_B r'$. Now by an obvious substitution lemma we have
	$\subst t x r \tto_{AB} \subst t x {r'}$, hence the following diagram commutes.
	\[\xymatrix{
		(\lm x.t)r \ar_A[d] \ar^B[r] 
			& (\lm x.t)r' \ar^A[d] \\
		\subst t x r \ar_{AB}@{->>}[r] 
			& \subst t x {r'}
	}\]
\item Let $\cctx {E^s} {\mu\alpha.[\alpha]\mu\beta.c} 
		\to_B \cctx {E^s} {\mu\alpha.\subst c \beta {\alpha\ \Box}}
		\to_A \mu\alpha.\subst {\subst c \beta {\alpha\ \Box}} \alpha {\alpha E^s}$.
	Now the following diagram commutes by an obvious substitution lemma.
	\[\xymatrix{
		\cctx {E^s} {\mu\alpha.[\alpha]\mu\beta.c} \ar_A[d] \ar^B[rr] 
			& & \cctx {E^s} {\mu\alpha.\subst c \beta {\alpha\ \Box}} \ar^A[d] \\
		\mu\alpha.[\alpha]\cctx {E^s} {\subst {\mu\beta.c} \alpha {\alpha E^s}} \ar_A[rd] 
			& 
			& \mu\alpha.\subst {\subst c \beta {\alpha\ \Box}} \alpha {\alpha E^s}\\
		 & 
		 	{\quad}\save[]-<0cm,0.1cm>*{ \mu\alpha.[\alpha]\mu\beta.\subst {\subst c \alpha {\alpha E^s}} \beta {\beta E^s} } \restore \ar_B[ru]
			 &
	}\]
	\qedafterarray
\end{enumerate}
\end{proof}

\begin{corollary}
\label{corollary:muT_sn_multiple_postponement}
A single $\to_A$-reduction step after multiple $\to_B$-reduction steps can be advanced.
That means, if $t_1 \tto_B t_2 \to_A t_3$, then there is a $t_4$ such that the 
following diagram commutes.
\[\xymatrix{
	t_1 \ar_A@{.>}[d] \ar^B@{->>}[r] & t_2 \ar^A@{->}[d]\\
	t_4 \ar_{AB}@{.>>}[r] & t_3
}\]	
\end{corollary}
 
\begin{proof}
The result holds by repeatedly applying Lemma~\ref{lemma:muT_sn_postponement} 
starting from right to left as the diagram indicates.
\[\xymatrix{
	t_1 \ar_A@{.>}[d]\ar^B@{->}[r]
		& t_2 \ar_A@{.>}[d] \ar^B@{->>}[rr] &
		& t_{n-1} \ar_A@{.>}[d]\ar^B@{->}[r]
		& t_n \ar_A@{->}[d] \\
	t_1'  \ar_{AB}@{.>>}[r]
		& t_2'  \ar_{AB}@{.>>}[rr] 	& 
		& t_{n-1}'  \ar_{AB}@{.>>}[r]
		& t_n' 
}\]
\qedafterarray
\end{proof}

\begin{lemma}
\label{theorem:muT_snA_snAB}
If $t \in \SNmuT[A]$, then $t \in \SNmuT[AB]$.
\end{lemma}

\begin{proof}
We prove this result by induction on the derivation of $t \in \SNmuT[A]$, so
by the induction hypothesis we obtain that for each term $t'$ with $t \to_A t'$
we have $t' \in \SNmuT[AB]$. By Lemma~\ref{lemma:muT_snB} we have 
$t \in \SNmuT[B]$, hence it suffices to prove that for all reduction sequences 
$t \tto_B t_2 \to_A t_3$ we have $t_3 \in \SNmuT[AB]$. Now by 
Corollary~\ref{corollary:muT_sn_multiple_postponement} we obtain a $t_4$ such
that the following diagram commutes.
\[\xymatrix{
	t \ar_A@{.>}[d] \ar^B@{->>}[r] & t_2 \ar^A@{->}[d]\\
	t_4 \ar_{AB}@{.>>}[r] & t_3
}\]	
By the induction hypothesis we have $t_4 \in \SNmuT[AB]$. Therefore,
since $t_4 \tto_{AB} t_3$, we have $t_3 \in \SNmuT[AB]$ by 
Fact~\ref{fact:muT_steps_sn}, so we are done.
\end{proof}

\begin{theorem}
\index{Strong normalization!for $\lambdamuT$}
\label{theorem:muT_sn}
If $t$ is well-typed, then $t \in \SNmuT[AB]$.
\end{theorem}

\begin{proof}
This result follows directly from Theorem~\ref{theorem:muT_snA_snAB} and
Corollary~\ref{corollary:muT_snA}.
\end{proof}

%% file: conclusions.tex
\section{Conclusions and further work}
In this paper we have introduced the \lambdamucal{T}, an extension of Parigot's 
\lambdamucal{} to include a type of natural numbers $\natTypeT$ with 
primitive recursion $\nrecT$, \`a la \GodelsTfull{}. We have proven the main
meta-theoretical properties and have shown that exactly the provably recursive
functions in first-order arithmetic can be represented.

In order to maintain confluence and a normal form theorem the \lambdamucal{T} 
is not a straightforward combination of the \lambdamucal{} and \GodelsTfull{}.
Both these systems are originally \cbn{}, whereas \lambdamuT{} is 
a \cbn{} system with strict evaluation on datatypes.

In our treatment of the reduction rules in \lambdamuT{}, we
have observed a tension between the \cbn{} features taken
directly from Parigot's original calculus, and the need to restrict the
rules for the datatypes to be \cbv{}. We plan to investigate a
fully-fledged \cbv{} version of \lambdamuT{} (see for example~\cite{ong1997,py1998} for 
definitions of a \cbv{} variant of \lambdamu{}). We expect that, apart from 
our proof of strong normalization, most of our results will extend to such a 
system. For a proof of strong normalization we will likely experience problems 
related to those discussed in~\cite{david2005}. The key issue 
is our Lemma~\ref{lemma:muT_snA_lambda}, which states that if
$r \in \SNmuT$ and \mbox{$\cctx E {\subst t x r} \in \SNmuT$}, then 
$\cctx E {(\lm x.t)r} \in \SNmuT$. In a \cbv{} variant the reduction
rule $v(\mu \alpha.c) \to \mu \alpha.\subst c \alpha {\alpha\ (v\Box)}$
will complicate this because $(\lm x.t)r$ is not solely a $\beta$-redex anymore. 

Instead of the \lambdamucal{} it would be interesting to consider a
system with the control operators \catch{} and \throw{} as
primitive (see Figure~\ref{figure:catchthrow_typing} for the typing rules). 
Such a system is described by Crolard~\cite{crolard1999}, who
proves a correspondence with \lambdamu{}. 
Herbelin~\cite{herbelin2010} also considers a variant of such a system to
define an intuitionistic logic that proves a variant of Markov's principle. 

\begin{figure}[h!]
\centering
\subfloat[catch]{
	\AXC{$\mujudg \Gamma {\Delta, \alpha : \rho} t \rho$}
	\UIC{$\mujudg \Gamma \Delta {\catchin \alpha t} \rho$}
	\normalAlignProof
	\DisplayProof}
\subfloat[throw]{
	\AXC{$\mujudg \Gamma \Delta t \rho$}
	\AXC{$\alpha : \rho \in \Delta$}
	\BIC{$\mujudg \Gamma \Delta {\throwto t \alpha} \tau$}
	\normalAlignProof
	\DisplayProof}
\caption{The typing rules for the primitives \catch{} and \throw{}.}
\label{figure:catchthrow_typing}
\end{figure}

The further reaching goal would be to define a dependently
typed \lambdacal{} with datatypes and control operators that
allows program extraction from classical proofs. In such a calculus
one can write specifications of programs, which can then be proven
using classical logic. The extraction mechanism would then extract a
program from such a proof, where the classical reasoning steps are
extracted to control operators. This would yield programs-with-control
that are \emph{correct by construction} because they are extracted
from a proof of the specification. This would extend the well-known
extraction method for constructive proofs, see~\cite{paulin1989} for example, 
to classical proofs.

This goal is particularly useful to obtain provably correct algorithms
where the use of control operators would really pay off (for example if
a lot of backtracking is involved). See~\cite{caldwell2000} for 
applications to classical search algorithms. The work of 
Makarov~\cite{makarov2006} may also be useful here, as it gives ways to 
optimize program extraction to make it feasible for practical programming.

%% file: acknowledgments.tex
\paragraph{Acknowledgments}
We are grateful to the anonymous referees who spotted some mistakes in earlier 
versions of this paper and provided several helpful suggestions.